\documentclass[a4paper,reqno]{amsart}

\usepackage[T1]{fontenc}
\usepackage[english]{babel}
\usepackage{indentfirst}
\usepackage{amsmath}
\usepackage{amssymb}
\usepackage{amsfonts}
\usepackage{mathtools}
\usepackage{braket}
\usepackage{amsthm}
\usepackage{color}
\usepackage{graphicx}
\usepackage{caption}
\usepackage[autostyle, english=british]{csquotes}
\usepackage{bbold}
\usepackage{relsize}
\usepackage{bbm}
\usepackage{txfonts}
\usepackage{mathrsfs}
\usepackage{theoremref}
\usepackage{dsfont}
\usepackage{chngcntr}
\usepackage{apptools}
\usepackage{enumitem}
\usepackage[hidelinks]{hyperref}[backref]
\usepackage[capitalise]{cleveref}
\usepackage{yfonts}
\usepackage{float}
\usepackage{booktabs}

\numberwithin{equation}{section}

\newcommand{\abs}[1]{\left\lvert#1\right\rvert}
\newcommand{\norma}[1]{\left\lVert#1\right\rVert}
\newcommand{\normaop}[1]{{\left\vert\kern-0.1ex\left\vert\kern-0.1ex\left\vert #1 
    \right\vert\kern-0.1ex\right\vert\kern-0.1ex\right\vert}}

\crefformat{equation}{(#2#1#3)}
\DeclareMathOperator{\Realpart}{Re}
\renewcommand{\Re}{\Realpart}
\DeclareMathOperator{\Imaginarypart}{Im}
\renewcommand{\Im}{\Imaginarypart}

\DeclareMathOperator{\Arg}{Arg}
\DeclareMathOperator{\Res}{Res}

\showboxdepth=\maxdimen
\showboxbreadth=\maxdimen
\newcommand{\numberset}{\mathbb}
\newcommand{\N}{\numberset{N}}
\newcommand{\Z}{\numberset{Z}}
\newcommand{\R}{\numberset{R}}
\newcommand{\C}{\numberset{C}}

\theoremstyle{definition}
\newtheorem{definition}{Definition}[section]
\theoremstyle{plain}
\newtheorem{theorem}{Theorem}[section]
\newtheorem*{theorem*}{Theorem}

\theoremstyle{lemma}
\newtheorem{lemma}{Lemma}[section]
\theoremstyle{corollary}

\theoremstyle{proposition}
\newtheorem{proposition}{Proposition}[section]
\crefalias{proposition}{proposition}
\theoremstyle{remark}
\newtheorem{remark}{Remark}[section]
\newtheorem{innercustomgeneric}{\customgenericname}  
\providecommand{\customgenericname}{}
\newcommand{\newcustomtheorem}[2]{%
  \crefname{#2}{#2}{#2s}%
  \newenvironment{#1}[1]
  {%
   \renewcommand\customgenericname{#2}%
   \crefalias{innercustomgeneric}{#2}%
   \renewcommand\theinnercustomgeneric{##1}%
   \innercustomgeneric
  }
  {\endinnercustomgeneric}
}

\newcustomtheorem{customdefinition}{Definition}
\newcustomtheorem{customlemma}{Lemma}
\newcustomtheorem{customproposition}{Proposition}
\DeclareMathOperator{\ran}{Ran}
\DeclareMathOperator{\sgn}{sgn}
\DeclareMathOperator{\Span}{span}
\newcommand*\diff{\mathop{}\!\mathrm{d}}

\DeclareMathOperator{\supp}{supp}
\DeclareMathOperator{\diam}{diam}

\DeclareRobustCommand{\stirlingone}{\genfrac[]{0pt}{}}
\makeatletter

\title[Point interactions on domains]{Resonances and resonance expansions\\ for $\delta$-interactions on the half-space}

\author[D.Noja]{Diego Noja}
\address{Dipartimento di Matematica e Applicazioni, Universit\`a
 di Milano Bicocca,  via Roberto Cozzi 55, 20125 Milano, Italy}
\email{diego.noja@unimib.it}

\author[F.Raso Stoia]{Francesco Raso Stoia}
\address{Dipartimento di Matematica e Applicazioni, Universit\`a
 di Milano Bicocca,  via Roberto Cozzi 55, 20125 Milano, Italy}
\email{f.rasostoia@campus.unimib.it}

\begin{document}

\begin{abstract}

In this paper we describe the resonances of the singular perturbation of the Laplacian on the half space $\Omega =\mathbb R^3_+$ given by the self-adjoint operator named $\delta$-interaction. We will assume Dirichlet or Neumann boundary conditions on $\partial \Omega$. At variance with the well known case of $\R^3$, the resonances constitute an infinite set, here completely characterized. Moreover, we prove that resonances have an asymptotic distribution satisfying a modified Weyl law and we give the semiclassical asymptotics. Finally we give applications of the results to the asymptotic behavior of the abstract wave and Schr\"odinger dynamics generated by the Laplacian with a point interaction on the half-space. 
\end{abstract}

\maketitle

\begin{footnotesize}
 
 
\end{footnotesize}
\section{Introduction}
Schr\"odinger operators $H$ on unbounded domains, besides a rich and much developed spectral theory, are characterized by the existence of resonances, or more precisely {\em continuation resonances}, by which we mean the poles of the meromorphic continuation in the lower complex half-plane (in suitable topologies) of the resolvent $(H-z^{2})^{-1}$ (see the seminal paper \cite{Vainberg} for the rigorous introduction of the concept). Loosely speaking, resonances behave as a sort of complex eigenvalues; associated to them there are resonance functions, solving formally the eigenvalue equation but only locally square integrable. 
 Among the many reasons, resonances are of interest because the solution of the wave or Schr\"odinger equation propagating on non-compact domains can be expressed, asymptotically for large times and for initial data localized in space, as an expansion based on resonance functions \cite{AH-K:ResonanceExpansionGreenFunctionSchrodingerWave, LP:ScatteringTheory, DZ:MathematicalTheoryScatteringResonances}, as an alternative to the  expansion based on the spectral theorem. For these and other reasons, it is desirable to have as precise information as possible about the location and properties of resonances of self-adjoint operators. While a lot of rigorous results have been accumulated along the years about general qualitative properties of resonances and their several applications in Physics, Dynamical Systems and Geometry  (see the treatise \cite{DZ:MathematicalTheoryScatteringResonances} for an encyclopedic and recent account), precise and quantitative information also for specific classes of operators is however rare. In this paper we want to begin the analysis of the continuation resonances of the class of self-adjoint singular perturbations of the Laplacian called  $\delta$-interactions or point interactions on subdomains of $\R^n, \ n=2,3$. There is a huge literature on point interactions (see \cite{Albeverio:SolvableModelsQM} for the bibliography until 2005, still growing). Their interest is due mostly to the fact that, in many cases, they represent exactly solvable models in Quantum Mechanics and that they are suitable scaling limit of more realistic, non solvable models. However, they have been studied almost exclusively on the whole space $\R^n, \ n=1,2,3$. Only a handful of papers treats the case of bounded domains (see \cite {BFM:PointInteractionBoundedDomains, EM:OptimizationPrincipalEigenvalue1PointInteractionBounded}) of hyperbolic space (see the recent \cite{Derezinski24}, including also point interactions on spheres), while the case of unbounded domains not coinciding with $\R^3$ has been treated essentially only in the geometries of infinite tubes and layers, especially relevant to the subject of quantum waveguides (see the book \cite{EK2015}).  As regards resonances, as it is well known, the Laplacian with a single point interaction in $\R^n,\  n=1,2,3$ has at most a single resonance, a strong difference with typical Schr\"odinger operators $H=-\Delta+V$ (see for example \cite{CH2005}). On the contrary, $N$-point delta interactions have infinitely many resonances. This has been recognized originally in \cite{AH-K:PerturbationsResonances}, and later the distribution of complex resonances has been investigated again by Albeverio and Karabash \cite{Albeverio2020,AK:MultilevelAsymptoticDistributiojResonances,AK:ResonanceFreeRegionsandNonHermitianSpectralOptimization} and by Lipovsk\'{y} and Lotoreichik \cite{LL:AsymptoticsResonancesInducedPointInteractions}. The special case of real resonances in dimension 3 is discussed in \cite{MS:RealResonances3dSchrodingerPoint}. Properties of the resolvent in dimension two are instead studied in \cite{CMY:2dSchrodingerPointInteractionsThresholdExpansionsZeroModes}, but the analysis of resonances of point interactions in dimension two is still rather poor. The one dimensional case is the best known, see in particular the recent papers \cite{DMW24} and \cite{Sacchetti23} and references therein. We stress that the treatment of the one dimensional case is greatly simplified by the fact that the $\delta$ distribution can be considered as a form bounded perturbation of the Laplacian and can be treated, somewhat formally but safely, as a standard potential. We mention in one dimension the half-line with a delta potential and a Dirichlet boundary conditions, which is called the Winter model in the physical literature and has received some attention in \cite{DM:SemiclassicalResonanceAsymptoticDeltaHalfLine, Sacchetti23} (see also the references therein). This is by no means the case in higher dimensions, where the point interaction cannot be properly interpreted as a kind of Schr\"odinger operator "$-\Delta+\delta$". 
As regards resonances, little is known when the domain considered is an unbounded subset of $\R^{n}$. In the higher dimensional case we refer to some examples concerning waveguides (see \cite{EK2015}).

Here we consider the case of the half-space with Dirichlet or Neumann boundary conditions on the boundary, and a point interaction at a generic position $y$. The plan of the paper is the following. In Section 2 point interactions on the half-space are introduced, giving their complete definition and relevant properties. Because a rather standard construction is involved in the rigorous definition, we postpone the proofs of Proposition \ref{definition:onepointD} and Proposition \ref{definition:onepointN} to the Appendix. We stress that these model cases do not reduce to any combinations of point interactions in $\R^3$ (see Remark 2.5 and Remark 3.8). Then, in Section 3 we obtain the explicit computation of of all the resonances of the considered models, here with full details and relevant remarks. The explicit knowledge of the set of resonances allows us to give in Section 4 a result about the number of resonances in complex discs of radius  $R$ as $R\to\infty$ . It should be noticed that the  distribution does not resemble the analogous one for smooth radial potentials for the whole space. This different behavior was already noticed in \cite{LL:AsymptoticsResonancesInducedPointInteractions} for several point interactions in $\mathbb R^3$ and it is here confirmed for a domain not coinciding with $\R^3$ (see remark\ \ref{asympdim} for more details on this). In Section 5 we study the semiclassical asymptotic of resonances, which parallels the one dimensional case treated recently in \cite{DM:SemiclassicalResonanceAsymptoticDeltaHalfLine}. Finally, in Section 6 the expansion in resonances is deduced, both for the wave and Schr\"odinger evolution. More precisely, we give the expansion in resonances of the flow for the wave equation and of the propagator for the Schr\"odinger operator of a point interaction in the half-space.
The resonance expansion for point interactions was previously considered in models containing point interactions in \cite{AH-K:ResonanceExpansionGreenFunctionSchrodingerWave}, where an array of point interactions in $\R^3$ is studied, and in \cite{BNP05}, where a quantized model of matter and radiation introduced in \cite{NP98,NP99} is studied. Finally, we notice that a possible alternative to the direct analysis of continuation resonances given here is by means of the black box scattering theory developed by Sjostrand and Zworski in \cite{Sjostrand:LectureResonances, SZ:ComplexScalingDistributionScatteringPoles}. Black box analysis of resonances and the extension of the results obtained in the present paper to domains more general than the half-space, as well as the analysis of the two dimensional case, seem to be interesting tasks and they will be the object of future research.

\section{Preliminaries: Point Interaction on the half-space}
\label{preliminaries}
In this section we introduce the definition of the Laplacian with a point interaction in the half-space. A  definition of this kind of operators suitable for rather general domains in $\mathbb R^3$ could be obtained by means of the theory of self-adjoint extensions of symmetric operators, but in the special case here studied a direct definition will suffice.\par We set $\R_{+}^{3}=\Set{x=(x_{1}, x_{2}, x_{3})\in\R^{3}| x_{3}>0}$ and its boundary is $\partial \R^3_+=\left\{x=(x_{1}, x_{2}, x_{3})\in\R^{3}\ |\ x_{3}=0\right\}$.  \\
We denote with $G_{z, y}^{0}$ the Green's function centered at $y\in \R^{3}$ for the operator $-\Delta-z^{2}$ on the whole space
\begin{equation*}
G_{z, y}^{0}(x)=\frac{e^{iz\abs{x-y}}}{4\pi\abs{x-y}}\ ,\quad \quad \quad \Im z>0\ .
\end{equation*}
Notice that $\Im z>0$ implies $z^2\in \C\setminus [0,+\infty)$, the resolvent set of the free Laplacian on $\R^3$.\par
The Green function centered at $y\in\R^3_+$ of the operator $-\Delta-z^{2}$ on the half-space $\R_{+}^{3}$ with Dirichlet or Neumann boundary conditions is given respectively by
\begin{equation*}
G_{z, y}^{\text{D}}=G_{z, y}^{0}-h_{z, y}^{\textup{D}}\quad\textup{or}\quad G_{z, y}^{\textup{N}}=G_{z, y}^{0}-h_{z, y}^{\textup{N}},
\end{equation*}
where $h_{z, y}^{\textup{D}}$ and $h_{z, y}^{\textup{N}}$ solve in the classical sense respectively the following boundary value problems
\begin{equation*}
\begin{cases}
(-\Delta-z^{2})h_{z, y}^{\textup{D}}=0 \quad & \textup{in}\ \R_{+}^{3}\\
h_{z, y}^{\textup{D}}=G_{z, y}^{0} \quad   &\textup{in}\ \partial\R^3_+
\end{cases}\quad\text{or}\quad 
\begin{cases}
(-\Delta-z^{2})h_{z, y}^{\textup{N}}=0 \quad & \textup{in}\,\R_{+}^{3}\\
\frac{\partial h_{z, y}^{\textup{N}}}{\partial x_{3}}=\frac{\partial G_{z, y}^{0}}{\partial x_{3}} \quad   &\textup{in}\, \partial\R^3_+
\end{cases}.
\end{equation*}
By a simple symmetry argument one obtains that
\begin{equation*}
h_{z, y}^{\textup{D}}(x)=\frac{e^{iz\abs{x-\overline{y}}}}{4\pi\abs{x-\overline{y}}}\quad\textup{and}\quad h_{z, y}^{\textup{N}}(x)=-\frac{e^{iz\abs{x-\overline{y}}}}{4\pi\abs{x-\overline{y}}}, \ \ \ \ \ \ \Im z>0
\end{equation*}
where $\overline{y}=(y_{1}, y_{2}, -y_{3})$ denotes the reflection of the point $y$ with respect to the plane $x_3=0$.\par 
Moreover we will denote with $-\Delta^{\textup{D}}$ and $-\Delta^{\textup{N}}$ the free Laplacian on $\R_{+}^{3}$ with Dirichlet and Neumann boundary conditions respectively on $\partial\R^3_+$, considered as self-adjoint operators on $L^2(\R^3_+)$. These operators have the respective domains

\begin{align*}
D(-\Delta^{\textup{D}})&=H^{2}\left(\R_{+}^{3}\right)\cap H^{1}_{0}\left(\R_{+}^{3}\right) \\
D(-\Delta^{\textup{N}})&=\left\{\psi\in H^{2}(\R_{+}^{3}):\left.\frac{\partial\psi}{\partial x_{3}}\right\rvert_{x_{3}=0}=0\right\}.
\end{align*}
where in the Neumann case the boundary value of the partial derivative has to be intended in the sense of trace (see for closely related framework \cite{BRS:eigenvalueinequalities} and for the specific case of the half-space \cite{BS:QuasiBoundary}).
Let us fix an arbitrary $y\in \R_{+}^3$, the place where the singular interaction is placed. We consider the symmetric non self-adjoint restrictions of the Dirichlet and Neumann Laplacians 
\begin{align*}
-\mathring\Delta_y^{\textup{D}}& \ \ \text{with domain}\ \ D(-\mathring\Delta_y^{\textup{D}})=\left\{\psi\in D(-\Delta^{\textup{D}}), \ \psi(y)=0 \right\}\\
-\mathring\Delta_y^{\textup{N}}& \ \ \text{with domain} \ \ D(-\mathring\Delta_y^{\textup{N}})=\left\{\psi\in D(-\Delta^{\textup{N}}), \ \psi(y)=0 \right\}.
\end{align*}
These restrictions are well defined because, by Sobolev embedding, domain elements admit continuous representatives. Moreover, the restrictions are closed symmetric operators. It can be proven (see the Appendix) that they both have defect indices $(1,1)$. So, by the von Neumann-Kre\u{\i}n theory, they admit a one parameter family of self-adjoint extensions. Their definition and properties are given in the two following propositions, the proof of which is an adaptation of the analogous proof valid for the case of $\R^3$ or for the case of bounded domains, and it is reported in detail in the Appendix.


\begin{proposition}
\label{definition:onepointD}
Let $\alpha\in \mathbb R^*:=\mathbb R\cup\{\infty\}$ and $-\Delta_{\alpha, y}^{\textup{D}}$ be the one parameter family of self-adjoint extensions of the closed symmetric non self-adjoint operator $-\mathring\Delta_y^{\textup{D}}$ with domain $D(-\mathring\Delta_y^{\textup{D}})=\Set{\psi\in L^2(\R^3_+) | \psi\in D(-\Delta^{\textup{D}}), \ \psi(y)=0}$.
Then, the following representation holds for the operator domain, operator action and resolvent of $-\Delta_{\alpha, y}^{\textup{D}}$ (where it is understood that $\Im z>0$):
\begin{align}\label{dirichletoperator}
D\Bigl(-\Delta_{\alpha, y}^{\textup{D}}\Bigr)&=\Set{\psi\in L^{2}\Bigl(\R_{+}^{3}\Bigr)|\,\exists q\in \C \, :\ \psi=\varphi^{z}+qG_{z, y}^{\textup{D}}, \, \varphi^{z}\in D(-\Delta^{\textup{D}}),\, \varphi^{z}(y)=\Gamma_{\alpha, y}^{\textup{D}}q}\nonumber\\
-\Delta_{\alpha, y}^{\textup{D}}\psi&=-\Delta\varphi^{z}+z^{2}qG_{z, y}^{\textup{D}}\nonumber\\
(-\Delta_{\alpha, y}^{\textup{D}}-z^{2})^{-1}&=(-\Delta^{\textup{D}}-z^{2})^{-1}+(\Gamma_{\alpha, y}^{\textup{D}})^{-1}\langle \overline{G_{z, y}^{\textup{D}}}, \cdot\rangle_{L^{2}\left(\R_{+}^{3}\right)}G_{z, y}^{\textup{D}}\\
\Gamma_{\alpha, y}^{\textup{D}}(z)&=\alpha-\frac{iz}{4\pi}+h_{z, y}^{\textup{D}}(y)=\alpha-\frac{iz}{4\pi}+\frac{e^{iz\abs{y-\overline{y}}}}{4\pi\abs{y-\overline{y}}}=\alpha-\frac{iz}{4\pi}+\frac{e^{2iy_{3}z}}{8\pi y_{3}}\nonumber ,
\end{align}
\end{proposition}
\begin{remark}
\label{remark:onepointD}
In particular, $\alpha=\infty$ corresponds to $-\Delta^{\textup{D}}$.
\end{remark}
\begin{proposition}
\label{definition:onepointN}
Let  $\alpha\in \mathbb R^*:=\mathbb R\cup\{\infty\}$ and $-\Delta_{\alpha, y}^{\textup{N}}$ be the one parameter family of self-adjoint extensions of the closed symmetric non self-adjoint operator $-\mathring\Delta_y^{\textup{N}}$ with domain $D(-\mathring\Delta_y^{\textup{N}})=\Set{\psi\in L^2(\R^3_+)|\psi\in D(-\Delta^{\textup{N}}), \ \psi(y)=0}$.
Then, the following representation holds for the operator domain, operator action and resolvent of $-\Delta_{\alpha, y}^{\textup{N}}$ (where it is understood that $\Im z>0$):
\begin{align}
\label{neumannoperator}
D\Bigl(-\Delta_{\alpha, y}^{\textup{N}}\Bigr)&=\Set{\psi\in L^{2}\Bigl(\R_{+}^{3}\Bigr)|\,\exists q\in \C\, :\ \psi=\varphi^{z}+qG_{z, y}^{\textup{N}}, \, \varphi^{z}\in D(-\Delta^{\textup{N}}),\, \varphi^{z}(y)=\Gamma_{\alpha, y}^{\textup{N}}q}\nonumber\\
-\Delta_{\alpha, y}^{\textup{N}}\psi&=-\Delta\varphi^{z}+z^{2}qG_{z, y}^{\textup{N}}\nonumber\\
(-\Delta_{\alpha, y}^{\textup{N}}-z^{2})^{-1}&=(-\Delta^{\textup{N}}-z^{2})^{-1}+(\Gamma_{\alpha, y}^{\textup{N}})^{-1}\langle \overline{G_{z, y}^{\textup{N}}}, \cdot\rangle_{L^{2}\left(\R_{+}^{3}\right)}G_{z, y}^{\textup{N}}\\
\Gamma_{\alpha, y}^{\textup{N}}(z)&=\alpha-\frac{iz}{4\pi}+h_{z, y}^{\textup{N}}(y)=\alpha-\frac{iz}{4\pi}-\frac{e^{iz\abs{y-\overline{y}}}}{4\pi\abs{y-\overline{y}}}=\alpha-\frac{iz}{4\pi}-\frac{e^{2iy_{3}z}}{8\pi y_{3}}\nonumber
\end{align}
\end{proposition}
\begin{remark}
\label{remark:onepointN}
In particular, $\alpha=\infty$ corresponds to $-\Delta^{\textup{N}}$.
\end{remark}
\begin{remark}
The generic domain element of a point interaction is decomposed in the sum of a regular part ($\varphi^z$) belonging to the domain of the Laplacian with relevant boundary conditions and a singular part ($qG_{z, y}^{\textup{D}}$ or $qG_{z, y}^{\textup{N}}$) proportional to the Green's function for the given domain and boundary condition. The decomposition is complemented by a relation linking the regular part at $y$ and the coefficient $q$ of the singular part, playing the role of a further boundary condition at the singular point $y$.  For Dirichlet or Neumann case, the boundary conditions at the singularity are given by $\varphi^{z}(y)=\Gamma_{\alpha, y}^{\textup{D}}q$ or $\varphi^{z}(y)=\Gamma_{\alpha, y}^{\textup{N}}q$ respectively. The latter contains the parameter $\alpha$ which specifies the particular self-adjoint extension of the free Laplacian among the whole family of its self-adjoint extensions. 
This pattern is typical of point interactions in their various incarnations. We add that the physical interpretation of $\alpha$ is related to the scattering length of the interaction.
\end{remark}

\section{Resonances of point interactions in the half-space}

A first result of this work concerns the localization of the continuation resonances for the operators $-\Delta_{\alpha, y}^{\textup{D}}$ and $-\Delta_{\alpha, y}^{\textup{N}}$. To introduce the concept of resonances for these perturbations of the Laplacian we first recall the corresponding definition for the free Laplacian on the half space. The resolvent $\left(-Delta^{\textup{D}}-z^{2}\right)$ is defined as follows (we consider the Dirichlet case)
\begin{gather*}
(-\Delta^{\textup{D}}-z^{2})^{-1}: L^{2}(\R_{+}^{3})\to H^{2}(\R_{+}^{3}), \quad\Im{z}>0\\
\left[(-\Delta^{\textup{D}}-z^{2})^{-1}f\right](x)=\int_{\R_{+}^{3}}G_{z, y}^{\textup{D}}(x)f(y)\diff y.
\end{gather*}
This is defined only for $\Im{z}>0$ because the kernel $G_{z, y}^{\textup{D}}(x)$ fails to be in $L^{2}(\R_{+}^{3})$ for $\Im{z}\le 0$ since it becomes non square integrable in a neighborhood of infinity. If we restrict its domain only to functions with compact support, this behavior of the kernel ceases to be problematic:
\begin{equation}
(-\Delta^{\textup{D}}-z^{2})^{-1}: L_{\textup{comp}}^{2}(\R_{+}^{3})\to H_{\textup{loc}}^{2}(\R_{+}^{3}).
\end{equation}
The issue is that $L_{\textup{comp}}^{2}(\R_{+}^{3})$ and $H_{\textup{loc}}^{2}(\R_{+}^{3})$ (or also $L_{\textup{loc}}^{2}(\R_{+}^{3})\supset H_{\textup{loc}}^{2}(\R_{+}^{3})$) are not Banach spaces and to introduce in the simplest way the procedure of holomorphic continuation, a common and useful trick consists in bracketing the resolvent by a smooth cut-off function $\rho$,  so obtaining the so called truncated resolvent
\begin{equation}
\left[\rho(-\Delta^{\textup{D}}-z^{2})^{-1}\rho f\right](x)=\rho(x)\int_{\Omega}G_{z, y}^{\textup{D}}(x)\rho(y)f(y)\diff y.
\end{equation}
This is a map between Banach spaces: $\rho(-\Delta^{\textup{D}}-z^{2})^{-1}\rho: L^{2}(\R_{+}^{3})\to L^{2}(\R_{+}^{3})$ for every cut-off function $\rho$ that turns out to be well defined and holomorphic $\forall z\in\C$. So, by definition, the free resolvent can be holomorphically continued to $\mathbb C$ if the truncated resolvent can be holomorphically continued to $\mathbb C$ as a map $\rho(-\Delta^{\textup{D}}-z^{2})^{-1}\rho: L^{2}(\R_{+}^{3})\to L^{2}(\R_{+}^{3})$.

Now we consider the single point interactions on $\R_{+}^{3}$. We want to carry out the same extension as done for the free Laplacian. Looking at the action of the resolvent in formula \cref{dirichletoperator} we observe that:
\begin{itemize}
\item[-]The first term is simply the free resolvent, which can be holomorphically extended to the whole $\C$ as recalled above.
\item[-] In the second term, the projection onto the function $G_{z, y}^{\textup{D}}$, viewed as an operator $L_{\textup{comp}}^{2}(\R_{+}^{3})\to L_{\textup{loc}}^{2}(\R_{+}^{3})$, can be holomorphically extended to all $\C$ in the sense defined above.
\item[-] In the second term, the factor $(\Gamma_{\alpha, y}^{\textup{D}}(z))^{-1}$ in front of the projector, is well defined out of the zeroes of $\Gamma_{\alpha, y}^{\textup{D}}(z)$.
\end{itemize}
The same procedure can be performed for $-\Delta_{\alpha, y}^{\textup{N}}$. 
Summarizing we can state the following proposition and definitions (here $H$ denotes either $-\Delta_{\alpha, y}^{\textup{D}}$ or $-\Delta_{\alpha, y}^{\textup{N}}$ and $\Gamma_{\alpha, y}$ either $\Gamma_{\alpha, y}^{\textup{D}}$ or $\Gamma_{\alpha, y}^{\textup{N}}$).
\begin{proposition}
The map
\begin{equation*}
z\mapsto(H-z^{2})^{-1}:  L_{\textup{comp}}^{2}(\R_{+}^{3})\to L_{\textup{loc}}^{2}(\R_{+}^{3})
\end{equation*}
extends as a meromorphic family of operators to $z\in\C$. 
\end{proposition}
\begin{definition}\label{defres}
We call {\em resonances} the poles of the meromorphic continuation of $H$, or explicitly, the solutions $z\in \mathbb C$ of the equation
\begin{equation}
\label{polesequation}
\Gamma_{\alpha, y}(z)=0.
\end{equation}
\end{definition}
\begin{remark} The previous definition agrees with the treatment based on \cite{Albeverio:SolvableModelsQM} in the case of one or more point interactions. Due to the structure of the resolvent operator, it is also consistent with the treatment of Sj\"{o}strand and Zworski in \cite{SZ:ComplexScalingDistributionScatteringPoles} and \cite{Sjostrand:LectureResonances} (see also the treatise \cite{DZ:MathematicalTheoryScatteringResonances} by Dyatlov and Zworski). Notice that from the point of view of spectral theory the relevant object is $z^2$ and not directly $z$. In particular, a solution $z$ of \cref{polesequation} such that $\Im{z}\ge 0$ and the corresponding solution $\psi$ of $H\psi=z^{2}\psi$ is in $L^{2}(\R_{+}^{3})$ gives rise to an eigenvalue $z^2$ of $H$. Due to self-adjointness, such $z$ can only belong to the real axis or to the positive imaginary axis. Moreover, still due to self-adjointness, for $\Im z>0$ the only possible solutions to \cref{polesequation} correspond to eigenvalues.
A solution $z$ of \cref{polesequation} such that $\Im{z}\le 0$ and not corresponding to an eigenvalue, gives rise to a genuine continuation resonance of $H$. We prefer however for the sake of simplicity to retain the spectral parameter $z$ as the main object. According to \cref{defres} there is no distinction between eigenvalues and continuation resonances. In the following, however, we will always make clear when a solution of \cref{polesequation} gives rise to an eigenvalue or to a continuation resonance.
\end{remark}
\begin{remark}
We recall that the Dirichlet and Neumann Laplacian in the half-space have no resonances at all. 
On the other hand the point interactions $\Delta_{\alpha,y}$ on the whole $\mathbb R^3$ have, whatever be $y$, a single resonance, existing for every $\alpha\in \mathbb R$ and given in that case by $z=-4\pi i\alpha$. For $\alpha<0$ the resonance corresponds actually to the  eigenvalue $-(4\pi\alpha)^2$, while for $\alpha\geq 0$ is a genuine continuation resonance.
\end{remark}

\begin{definition}\label{multiplicitydef}
The multiplicity of a resonance $z$ of the operator $H$ is the algebraic multiplicity of the corresponding solution of \cref{polesequation}.
\end{definition}

\begin{remark} 
In particular, the resonances (eigenvalues or continuation resonances) are always simple for the operators here treated apart for the Dirichlet case when $\alpha=-1/(8\pi y_{3})$, where the multiplicity is two. For a resonance $z$ to have multiplicity greater than one it must hold that
\begin{equation*}
\begin{cases}
\Gamma_{\alpha, y}(z)=0\\
\frac{\partial\Gamma_{\alpha, y}}{\partial z}(z)=0
\end{cases}
\end{equation*}
or explicitly
\begin{equation*}
\begin{cases}
\alpha-\frac{iz}{4\pi}\pm\frac{e^{2iy_{3}z}}{8\pi y_{3}}=0\\
\pm e^{2iy_{3}z}=1
\end{cases}
\end{equation*}
where the plus sign denotes the Dirichlet case and the minus sign the Neumann one. 

The second equation of the Dirichlet system has solutions $z_{\textup{D}}=k\pi/y_{3}$, $k\in\Z$. Substituting it in the first equation leads to the condition
\begin{equation*}
\alpha-\frac{ik}{4y_{3}}+\frac{1}{8\pi y_{3}}=0,
\end{equation*}
which is verified if and only if $k=0$ (hence $z=0$) and $\alpha=-1/(8\pi y_{3})$.

The second equation of the Neumann system has instead solutions $z_{\textup{N}}=\pi/(2y_{3})+k\pi/y_{3}$, $k\in\Z$. Substitution easily shows that for thew system to be solvable it needs to be $k=-1/2$ and hence there cannot be resonances with multiplicity grater than one in the Neumann case. Finally, noticing that $\frac{\partial^{2}\Gamma_{\alpha, y}}{\partial z^{2}}(z)\ne 0$ for all $z\in\C$ we conclude that the multiplicity of a resonance is at most two.
\end{remark}
The spectrum of point interactions $-\Delta_{\alpha, y}^{\textup{D}}$ and $-\Delta_{\alpha, y}^{\textup{N}}$ is given in the following proposition. Other properties about the spectrum of point interaction in more general unbounded sub-domains of $\mathbb R^2$ and $\mathbb R^3$ will appear in the forthcoming article \cite{NR-S24} by these authors. 
\begin{proposition}
\label{thspectrumhalfspacedirichlet}{}
The spectra of the operators $-\Delta_{\alpha, y}^{\textup{D}}$ and $-\Delta_{\alpha, y}^{\textup{N}}$ are characterized as follows. 
\begin{itemize} 
\item[a) ]$\sigma_c(-\Delta_{\alpha, y}^{\textup{D}})=\sigma_{ac}(-\Delta_{\alpha, y}^{\textup{D}})=[0,+\infty)$. Moreover,  $\sigma_{p}(-\Delta_{\alpha, y}^{\textup{D}})$ is empty or it consists of a single non positive simple eigenvalue existing if and only if $\alpha<-\frac{1}{8\pi y_{3}}$.
\label{thspectrumhalfspaceneumann}{}
\item[b)] $\sigma_c(-\Delta_{\alpha, y}^{\textup{N}})=\sigma_{ac}(-\Delta_{\alpha, y}^{\textup{N}})=[0,+\infty)$. Moreover,  $\sigma_{p}(-\Delta_{\alpha, y}^{\textup{N}})$ is empty or it consists of a single simple negative eigenvalue existing if and only if $\alpha<\frac{1}{8\pi y_{3}}$.
\end{itemize}
\end{proposition}
\begin{proof}
The general theory of self-adjoint extensions of positive symmetric operators with equal and finite deficiency indices assures that all the members of the family of self-adjoint extension share the same continuous spectrum, and that the number of negative eigenvalues of member of the family cannot exceed the deficiency index (counting multiplicity), which here is one (see for example \cite{AG63}, Section 85). The absence of singular continuous spectrum is a consequence of the explicit form of the resolvent given in formulae \cref{dirichletoperator} and \cref{neumannoperator} and of the direct application of the standard criterion in \cite{RS4}, Theorem XIII.19. The quantitative threshold values of the existence of eigenvalues are obtained by inspection of formulae \eqref{complexsystemhalfspace} with $a=0$ and $b\geq0$ as given in the proof of the subsequent Proposition \ref{DirSpec}, to which we refer for a similar analysis.
\end{proof}

The properties of resonances of the point interactions $-\Delta_{\alpha, y}^{\textup{D}}$ and $-\Delta_{\alpha, y}^{\textup{N}}$ are completely and explicitly described by the following propositions.
\begin{proposition}\label{realresonances}
Operators $-\Delta_{\alpha, y}^{\textup{D}}$ and $-\Delta_{\alpha, y}^{\textup{N}}$ have no continuation resonances on the real axis apart possibly $z=0$. Moreover, $z=0$ is a resonance if and only if $\alpha=-1/(8\pi y_{3})$ in the Dirichlet case and if and only if $\alpha= 1/(8\pi y_{3})$ in the Neumann case.
\end{proposition}
\begin{proof}
Let be $z\in\C$ such that $\Gamma_{\alpha, y}(z)=0$. Depending on the boundary condition, according to the form of $\Gamma_{\alpha, y}$ given in \cref{dirichletoperator} and \cref{neumannoperator} the conditions takes the form ($+$ sign Dirichlet, $-$ sign Neumann)
\begin{equation*}
\alpha-\frac{iz}{4\pi}\pm\frac{e^{2iy_{3}z}}{8\pi y_{3}}=0.
\end{equation*}
To investigate the existence of real resonances we substitute $z=a\in\R$ in the above equation. Separating the real and imaginary part leads to the following system
\begin{equation*}
\begin{cases}
\cos{(2y_{3}a)}=\mp 8\pi\alpha y_{3}\\
\sin{(2y_{3}a)}=\pm 2ay_{3}
\end{cases}.
\end{equation*}
Since $y_{3}>0$ the second equation admits only $a=0$ as possible solution. In this case, because of the first equation, it must be $\alpha=\mp 1/(8\pi y_{3})$.
\end{proof}

\begin{remark}
Distributional solutions (where distributions are meant in $\mathbb R^3_+$) of the equation $(H-z^{2})\psi_z=0$ where again $H$ denotes either $-\Delta_{\alpha, y}^{\textup{D}}$ or $-\Delta_{\alpha, y}^{\textup{N}}$, are called resonance functions corresponding  to the resonance $z$. Resonance functions belong to the space $H^2_{\text{loc}}(\mathbb R^3_+\setminus\{y\})$, satisfy the boundary conditions on the boundary plane $x_3=0$ and at the singularity $y$ (see remark 2.3), and for $\Im z<0$ they are exponentially growing at infinity. For $z=0$ it can be shown that resonance functions belong to certain weighted Sobolev spaces depending on the dimension and in particular they decay at infinity in dimension three (see \cite{Scandone:21} for the case of several point interactions, which in this respect is completely analogous to the present case, and the seminal paper \cite{JK:JK} for the case of Schr\"odinger operators). However, will not need these informations in the present analysis.
\end{remark}

\begin{remark}
If $z=0$ is a solution of \cref{polesequation} two possibilities arise. The first is that the resonance function $\psi_0$ is actually an eigenvector, so belonging to the operator domain and corresponding to a zero resonance which is also an eigenvalue. The second is that $\psi_0\notin L^2$, it is a genuine resonance function and in this case $\psi_0$ behaves at infinity as the fundamental solution of the Laplacian: $\psi_0\sim \frac{c}{|x-y|}$. These facts are explicit in subsequent formula \eqref{eigenvectorD} for the Dirichlet case where a threshold eigenvalue appears, and in the formula \eqref{thresholdN} for the Neumann case where a genuine threshold resonance appears.
\end{remark}

\begin{proposition}
\label{DirSpec}
Consider the operator $-\Delta_{\alpha, y}^{\textup{D}}$\ . Then the following properties hold:
\begin{itemize}
\item[i)] For all $\alpha\in\R$ and $y\in\R_{+}^{3}$, in every region of the form $\frac{k\pi}{y_{3}}<\abs{\Re{z}}<\frac{(2k+1)\pi}{2y_{3}}$, with $k=1, 2, \dots$, there are exactly two resonances $z$ for $-\Delta_{\alpha, y}^{\textup{D}}$ with opposite real part.\\ Moreover
\begin{equation}
\label{bintermsofahalfspace2}
\Im{z}=\frac{1}{2y_{3}}\ln{\Biggl(\frac{\sin{(2y_{3}\Re{z})}}{2y_{3}\Re{z}}\Biggr)}.
\end{equation}
\item[ii)] The value $z=0$ is a resonance if and only if $\alpha=-\frac{1}{8\pi y_{3}}$ and in that case it is an eigenvalue. The corresponding eigenvector is 
\begin{equation}\label{eigenvectorD}
\psi_0(x)=\frac{1}{4\pi\abs{x-{y}}}-\frac{1}{4\pi\abs{x-\overline{y}}}
\end{equation}
\item[iii)] There is a resonance $z$ on the negative imaginary semi-axis if and only if $\alpha<-\frac{1}{8\pi y_{3}}$.\\ When it exists this $z$ is unique.
\item[iv)] The region $\Re{z}\in\Bigl(-\frac{\pi}{2y_{3}}, \frac{\pi}{2y_{3}}\Bigr)$ contains resonances if and only if $\alpha>-\frac{1}{8\pi y_{3}}$. In such a case there are two resonances with opposite real part.
\item[v)] If $\alpha=\frac{1}{8\pi y_{3}}\ln{\Bigl(\frac{\pi}{2}+k\pi\Bigr)}$ for some non-negative even $k$, then there is a resonance in $z=\frac{\pi}{4y_{3}}+\frac{k\pi}{2y_{3}}-\frac{i}{2y_{3}}\ln{\Bigl(\frac{\pi}{2}+k\pi\Bigr)}$.
\item[vi)] If $\alpha=\frac{1}{8\pi y_{3}}\ln{\Bigl(-\frac{\pi}{2}-k\pi\Bigr)}$ for some negative odd $k$, then there is a resonance in $z=\frac{\pi}{4y_{3}}+\frac{k\pi}{2y_{3}}-\frac{i}{2y_{3}}\ln{\Bigl(-\frac{\pi}{2}-k\pi\Bigr)}$.
\end{itemize}
\end{proposition}

The point interaction with Neumann boundary conditions is treated in the following result.
\begin{proposition}
\label{NeuSpec}
Consider the operator $-\Delta_{\alpha, y}^{\textup{N}}$. Then the following properties hold:
\begin{itemize}
\item[i)] For all $\alpha\in\R$ and $y\in\R_{+}^{3}$, in every region of the form $\frac{(2k+1)\pi}{2y_{3}}<\abs{\Re{z}}<\frac{(k+1)\pi}{y_{3}}$ with $k=0, 1, \dots$, there are exactly two resonances $z$ for $-\Delta_{\alpha, y}^{\textup{N}}$ with opposite real part. Moreover
\begin{equation}
\label{bintermsofahalfspaceNeumann}
\Im{z}=\frac{1}{2y_{3}}\ln{\Biggl(-\frac{\sin{(2y_{3}\Re{z})}}{2y_{3}\Re{z}}\Biggr)}.
\end{equation}
\item[ii)] The value $z=0$ is a resonance if and only if  $\alpha=\frac{1}{8\pi y_{3}}$ and in that case it is not an eigenvalue. The corresponding resonance function is 
\begin{equation}\label{thresholdN}
\psi_0(x)=\frac{1}{4\pi\abs{x-{y}}}+\frac{1}{4\pi\abs{x-\overline{y}}}
\end{equation}

\item[iii)] There is a resonance $z$ on the negative imaginary semi-axis if and only if $\alpha>\frac{1}{8\pi y_{3}}$.\\ When it exists this $z$ is unique.
\item[iv)] If $\alpha=\frac{1}{8\pi y_{3}}\ln{\Bigl(\frac{\pi}{2}+k\pi\Bigr)}$ for some positive odd $k$, then there is a resonance in $z=\frac{\pi}{4y_{3}}+\frac{k\pi}{2y_{3}}-\frac{i}{2y_{3}}\ln{\Bigl(\frac{\pi}{2}+k\pi\Bigr)}$.
\item[v)] If $\alpha=\frac{1}{8\pi y_{3}}\ln{\Bigl(-\frac{\pi}{2}-k\pi\Bigr)}$ for some negative even $k$, then there is a resonance in $z=\frac{\pi}{4y_{3}}+\frac{k\pi}{2y_{3}}-\frac{i}{2y_{3}}\ln{\Bigl(-\frac{\pi}{2}-k\pi\Bigr)}$.
\end{itemize}
\end{proposition}

\begin{remark}The resonances on the negative imaginary axis, existing under the conditions stated in Proposition \ref{DirSpec}(iii) and \ref{NeuSpec}(iii) are usually called, especially in the physical literature, \emph{antibound states}.
\end{remark}
In the following we prove \cref{DirSpec}. The proof of \cref{NeuSpec} is not reported because, barring some minor details, it follow the same path of the Dirichlet case.\\ A plot of Dirichlet resonances with $\Re z>0$ is given in Figure 1.\\
We premit the following lemma, needed in the proof of part i) of the Proposition \ref{DirSpec}.
\begin{lemma}
\label{ghlemma}
Let us define
$$g(a)=\frac{\sin{(2y_{3}a)}}{2y_{3}a},\ \ \ \ h(a)=e^{-8\pi\alpha y_{3}-\frac{2y_{3}a}{\tan{(2y_{3}a)}}}\ .$$ \ 
Then 
\begin{itemize}
\item[i)] in each interval of the form $\Bigl(\frac{k\pi}{y_{3}}, \frac{(2k+1)\pi}{2y_{3}}\Bigr)$ with $k=1, 2, \dots$, the function $g$ is positive and has a unique maximum $M_{k}\in\Biggl(\frac{k\pi}{y_{3}}, \frac{\frac{\pi}{2}+2k\pi}{2y_{3}}\Biggr)$;
\item[ii)] $g$ is concave in every interval of the form $\Biggl(\frac{k\pi}{y_{3}}, \frac{\frac{\pi}{2}+2k\pi}{2y_{3}}\Biggr)$, $k=1, 2, \dots$.
\item[iii)] the function $h$ is monotone increasing in every interval of the form $\Bigl(\frac{k\pi}{y_{3}}, \frac{(2k+1)\pi}{2y_{3}}\Bigr)$ with $k=0, 1, 2, \dots$. When $k\ge 1$ it is also convex in the same intervals.
\end{itemize}
\end{lemma}
\begin{proof}
We only give the proof for $iii)$, the proof of $i)$ and $ii)$ being entirely elementary.
The derivative of $h$ is
\begin{equation*}
h'(a)=-2y_{3}e^{-8\pi\alpha y_{3}-\frac{2y_{3}a}{\tan{(2y_{3})}}}\frac{\tan{(2y_{3}a)}-2y_{3}a\Bigl(1+\tan^{2}{(2y_{3}a)}\Bigr)}{\tan^{2}{(2y_{3}a)}}.
\end{equation*}
We claim that $h'(a)\ge 0$ for all $ a\ge 0$ (where defined). This is a consequence of the elementary inequality
\begin{equation*}
p(t)=\frac{\tan{t}}{1+\tan^{2}{t}}\le t
\end{equation*}
being satisfied for all $t\ge 0$, with the equality only valid in $t=0$. The continuity of $h$ and its derivative in every interval of the form $\Bigl(\frac{k\pi}{y_{3}}, \frac{(2k+1)\pi}{2y_{3}}\Bigr)$ with $k=0, 1, 2, \dots$, then implies that $h$ is monotone increasing there. 
The second derivative $h''$ is given by
\begin{equation*}
h''(a)=4y_{3}^{2}e^{-8\pi\alpha y_{3}-\frac{2y_{3}a}{\tan{(2y_{3})}}}\Bigl(\cot^{2}{(2y_{3}a)}+2\csc^{2}{(2y_{3}a)}
+4y_{3}^{2}a^{2}\csc^{4}{(2y_{3}a)}-8y_{3}a\cot{(2y_{3}a)}\csc^{2}{(2y_{3}a)}\Bigr).
\end{equation*}
It holds that $h''>0$ when
$
t^{2}+\sin^{2}{t}\cos^{2}{t}+2\sin^{2}{t}-4t\sin{t}\cos{t}>0,
$
with $t=2y_{3}a$. This expression is greater or equal than $t^{2}-2t$, which is positive for $t>2$. So $h''(a)>0$ for all $a>1/y_{3}$ in its domain. Since $1/y_{3}<\pi/y_{3}\le k\pi/y_{3}$ for $k=1, 2, \dots$, then $h''>0$ in each interval $\Bigl(\frac{k\pi}{y_{3}}, \frac{(2k+1)\pi}{2y_{3}}\Bigr)$ with $k=1, 2, \dots$ and it is convex there.
\end{proof}

\begin{proof}[Proof of \cref{DirSpec}]
The definition of resonance and \cref{dirichletoperator} imply that the continuation resonances of $-\Delta_{\alpha, y}^{\textup{D}}$ are determined by the solutions of the equation $\Gamma_{\alpha, y}^{\textup{D}}(z)=0$ with $\Im{z}<0$. Hence, we set $z=a+ib$ in $\Gamma_{\alpha, y}^{\textup{D}}(z)=0$, with the prescription $b<0$. This way we get
\begin{equation*}
\alpha+\frac{b-ai}{4\pi}+e^{-2y_{3}b}\frac{\cos{(2y_{3}a)}+i\sin{(2y_{3}a)}}{8\pi y_{3}}=0.
\end{equation*}
Equating real and imaginary part of both sides, we obtain the system in $a$ and $b$
\begin{equation}
\label{complexsystemhalfspace}
\begin{cases}
-\alpha-\frac{b}{4\pi}=e^{-2y_{3}b}\frac{\cos{(2y_{3}a)}}{8\pi y_{3}}\\
\frac{a}{4\pi}=e^{-2y_{3}b}\frac{\sin{(2y_{3}a)}}{8\pi y_{3}}
\end{cases}.
\end{equation}

We start by looking for solutions on the negative imaginary semi-axis and we set $a=0$. We observe that in this case the second equation is satisfied for all $b$. The first equation of the system is then
\begin{equation*}
f(b)=\alpha+\frac{b}{4\pi}+\frac{e^{-2y_{3}b}}{8\pi y_{3}}=0,\quad b\leq0.
\end{equation*}
Since $\lim_{b\to -\infty}f(b)=+\infty$ and $f$ is a continuous monotonically decreasing function, a solution in $(-\infty, 0]$ exists only if $f(0)<0$. This happens when $\alpha\leq-(8\pi y_{3})^{-1}$, otherwise there are no resonances on the negative imaginary semi-axis. The possible solution is in fact unique because $f'(b)=(1-e^{-2y_{3}b})/(4\pi)<0$ for all $b<0$. This proves part iii) of the Proposition. In particular, one has $b=0$ if and only if  $\alpha=-(8\pi y_{3})^{-1}$. So, $z=0$ is a resonance when this condition is satisfied. In this case, it is easy to check that the resonance function 

\begin{equation}
\psi_0(x)=\frac{1}{4\pi\abs{x-{y}}}-\frac{1}{4\pi\abs{x-\overline{y}}}
\end{equation}
belongs to $L^2(\mathbb R^3_+)$; so it is an eigenvector and $z=0$ is an eigenvalue. This proves part ii).\\
Another case of interest in the study of \cref{complexsystemhalfspace} is when $a=\frac{k\pi}{2y_{3}}$, for some $k=\pm 1, \pm 2, \dots$. In this case it is easy to see that the system does not admit any solution, because the right side of the second equation of the system is zero, while the left one isn't. So there are no resonances on the lines $\Re{z}=\frac{k\pi}{2y_{3}}$, $k=\pm 1, \pm 2, \dots$ outside the real axis.

Also, since we already considered the $a=0$ case, we can rewrite the second equation of the system \cref{complexsystemhalfspace} as
\begin{equation*}
\frac{\sin{(2y_{3}a)}}{2y_{3}a}=e^{2y_{3}b}, \quad a\ne 0.
\end{equation*}
Notice that, as expected, because of the self-adjointness of the operator, the solutions with $\Im{z}\ge 0$ can only be found on the real axis or on the positive imaginary semi-axis.

Since the left side is always smaller than 1, a solution is possible only if $b<0$. It is immediate to see that there are no solution when $\pi+2k\pi<2y_{3}\abs{a}<2\pi+2k\pi$, $k=0, 1, 2, \dots$, because in that case the left side would be negative.

We now consider the cases when  $a=\frac{\pi}{4y_{3}}+\frac{k\pi}{2y_{3}}$ for some $k\in\Z$. From formula \cref{complexsystemhalfspace} we get
\begin{equation*}
\begin{cases}
b=-4\pi\alpha\\
e^{8\pi\alpha y_{3}}=(-1)^{k}\Bigl(\frac{\pi}{2}+k\pi\Bigr)
\end{cases}.
\end{equation*}
This system is overdetermined and so can have a solution only for certain values of $\alpha$ and $y_{3}$. It is easy to see that if $k$ is positive odd or negative even, then no solution exists regardless of $\alpha$ and $y_{3}$ ($a$ would be in one of the intervals mentioned above for which no solution exists); instead, 
if $k$ is non-negative even and $\alpha=\frac{1}{8\pi y_{3}}\ln{\Bigl(\frac{\pi}{2}+k\pi\Bigr)}$, then there is a resonance in $z=\frac{\pi}{4y_{3}}+\frac{k\pi}{2y_{3}}-\frac{i}{2y_{3}}\ln{\Bigl(\frac{\pi}{2}+k\pi\Bigr)}$; while, if $k$ is negative odd and $\alpha=\frac{1}{8\pi y_{3}}\ln{\Bigl(-\frac{\pi}{2}-k\pi\Bigr)}$, then there is a resonance in $z=\frac{\pi}{4y_{3}}+\frac{k\pi}{2y_{3}}-\frac{i}{2y_{3}}\ln{\Bigl(-\frac{\pi}{2}-k\pi\Bigr)}$. This proves part v) and vi).\par
After treating the exceptional cases, we turn to part $i)$ and $ii)$. As previously, we locate resonances thanks to the equation $\Gamma_{\alpha, y}^{\textup{D}}(z)=0$. This leads to \cref{complexsystemhalfspace}. We divide side by side the second equation for the first one. We notice that all the values leading to a vanishing denominator have been considered previously and hence performing this ratio does not cause loss of solutions. So, we get
$\tan{(2y_{3}a)}=-\frac{a}{4\pi\alpha+b}$,
which, solved for $b$ gives
\begin{equation*}
b=-4\pi\alpha-\frac{a}{\tan{(2y_{3}a)}}.
\end{equation*}
Substituting $b$ in the second equation of \cref{complexsystemhalfspace}, we get the following equation in $a$
\begin{equation}
\label{gandheqhalfspace}
g(a)=\frac{\sin{(2y_{3}a)}}{2y_{3}a}=e^{-8\pi\alpha y_{3}-\frac{2y_{3}a}{\tan{(2y_{3}a)}}}=h(a).
\end{equation}

Now we claim that in each interval of the form $\Bigl(\frac{k\pi}{y_{3}}, \frac{(2k+1)\pi}{2y_{3}}\Bigr)$ with $k=1, 2, \dots$, \cref{gandheqhalfspace} has exactly one solution.

Let $a_{k}$ be the smallest solution in the $k-th$ interval. We distinguish the two cases $a_{k}\ge M_{k}$ and $a_{k}<M_{k}$, where $M_{k}$ is the maximum of $g$ in $\Bigl(\frac{k\pi}{y_{3}}, \frac{(2k+1)\pi}{2y_{3}}\Bigr)$. If $a_{k}\ge M_{k}$, the solution is unique because in $\Bigl[a_{k}, \frac{(2k+1)\pi}{2y_{3}}\Bigr)$, $h$ is known to be increasing and $g$ is decreasing thanks to \cref{ghlemma} . If $a_{k}<M_{k}$, given that $M_{k}\in\Biggl(\frac{k\pi}{y_{3}}, \frac{\frac{\pi}{2}+2k\pi}{2y_{3}}\Biggr)$, the function $r(a)=h(a)-g(a)$ is strictly convex in $\Biggl(\frac{k\pi}{y_{3}}, M_{k}\Biggr)$  by the just above mentioned lemma. But this, together with $\lim_{a\to\frac{k\pi^{+}}{y_{3}}}r(a)=0$ implies that at there is at most a zero for $r$ in $\Biggl(\frac{k\pi}{y_{3}}, M_{k}\Biggr)$. In fact, if there were two internal zeros $a_{k}<\tilde{a}_{k}$, there would be two numbers $b_{k}$ and $\tilde{b}_{k}$ such that $\frac{k\pi}{y_{3}}<b_{k}<a_{k}<\tilde{b}_{k}<\tilde{a}_{k}<M_{k}$ and $r'(b_{k})=r'(\tilde{b}_{k})=0$. This contradicts the function $r$ being convex ($r'$ is strictly increasing); then there is a solution at most. 
For each interval, this possible unique solution in fact exists. This is true because the conditions 
\begin{equation}
\label{ghconditions}
h\Bigl(\frac{k\pi^{+}}{y_{3}}\Bigr)=g\Bigl(\frac{k\pi}{y_{3}}\Bigr)=0,\quad h'\Bigl(\frac{k\pi^{+}}{y_{3}}\Bigr)=0<g'\Bigl(\frac{k\pi}{y_{3}}\Bigr)=\frac{y_{3}}{k\pi},\quad h\Bigl(\frac{\pi^{-}+2k\pi}{2y_{3}}\Bigr)=+\infty\quad\textup{and}\quad g\Bigl(\frac{\pi^{-}+2k\pi}{2y_{3}}\Bigr)=0
\end{equation}
ensure that $r$ changes sign throughout the interval.

It remains to consider the interval $\Bigl[0, \frac{\pi}{2y_{3}}\Bigr)$. Here $g$ is monotone decreasing and $h$ is monotone increasing, so, at most there is one solution. The solution actually exists only if $\lim_{a\to 0}r(a)<0$, which happens if and only if $\alpha>-\frac{1}{8\pi y_{3}}$. This proves part $ii)$.

For each solution of \cref{gandheqhalfspace}, we can recover the corresponding imaginary part through the relation $b=\frac{1}{2y_{3}}\ln{g(a)}$, which corresponds to formula \eqref{bintermsofahalfspace2}. This finally proves part $i)$ of the proposition. 
\end{proof}
\begin{remark}
It is well known (see \cite{JK:JK}) that for Schr\"odinger operators $-\Delta +V$, zero energy resonances and zero energy eigenvalues can be distinguished by means of the different low energy behavior of the resolvent. Namely, the resolvent $(-\Delta + V-z^2)^{-1}$ behaves as $\frac{c_{-2}}{z^2}$ at a zero energy eigenvalue and as $\frac{c_{-1}}{z}$ at a zero energy resonance, for certain   operators $c_{-1}, c_{-2}$ that are bounded in suitable operator topologies. The same is true in the present example, not covered by the standard theory.
Let us consider the Dirichlet case. 
At the critical value $\alpha_{\textup{c}}=-\frac{1}{8\pi y_{3}}$ one has
\begin{equation*}
\begin{split}
\biggl(\Gamma_{\alpha, y}^{\textup{D}}(z)\biggr)^{-1} & = \Biggl(-\frac{1}{8\pi y_{3}}-\frac{iz}{4\pi}+\frac{e^{2iy_{3}z}}{8\pi y_{3}}\Biggr)^{-1}
                                                                                                            = \Biggl(-\frac{1}{8\pi y_{3}}-\frac{iz}{4\pi}+\frac{1+2iy_{2}z-\frac{1}{2}4y_{3}^{2}z^{2}-\frac{1}{6}8iy_{3}^{3}z^{3}}{8\pi y_{3}}+O(z^{4})\Biggr)^{-1}\\
                                                                                                           & = \Biggl(-\frac{y_{3}}{4\pi}z^{2}-\frac{iy_{3}^{2}}{6\pi}z^{3}+O(z^{4})\Biggr)^{-1}= -\frac{4\pi}{y_{3}}z^{-2}\Biggl(1+\frac{2iy_{3}}{3}z+O(z^{2})\Biggr)^{-1}
                                                                                                           = -\frac{4\pi}{y_{3}}z^{-2}+\frac{8}{3}\pi iz^{-1}+O(1).
\end{split}
\end{equation*}

For the Neumann case,
the critical value is now $\alpha_{\textup{c}}=\frac{1}{8\pi y_{3}}$. So
\begin{equation*}
\begin{split}
\biggl(\Gamma_{\alpha, y}^{\textup{N}}(z)\biggr)^{-1} & = \Biggl(\frac{1}{8\pi y_{3}}-\frac{iz}{4\pi}-\frac{e^{2iy_{3}z}}{8\pi y_{3}}\Biggr)^{-1}
                                                                                                             = \Biggl(\frac{1}{8\pi y_{3}}-\frac{iz}{4\pi}-\frac{1+2iy_{3}z-\frac{1}{2}4y_{3}^{2}z^{2}-\frac{1}{6}8iy_{3}^{3}z^{3}}{8\pi y_{3}}+O(z^{4})\Biggr)^{-1}\\
                                                                                                            & = \Biggl(-\frac{iz}{2\pi}+\frac{y_{3}}{4\pi}z^{2}+\frac{iy_{3}^{2}}{6\pi}z^{3}+O(z^{4})\Biggr)^{-1}= \frac{2\pi i}{z}\Biggl(1+\frac{iy_{3}}{2}z-\frac{y_{3}^{2}}{3}z^{2}+O(z^{3})\Biggr)^{-1}
                                                                                                             = 2\pi i z^{-1}+\pi y_{3}+ O(1).
\end{split}
\end{equation*}
\end{remark}
\begin{remark}{[\emph{Bifurcation of eigenvalues and resonances at zero energy}]} Two different pictures show up in the Dirichlet and Neumann case as regards the bifurcation of eigenvalues and resonances at the origin of the complex plane. In the Neumann case, things are simple: a resonance (an anti-bound state more precisely) moves on the negative imaginary axis when $\alpha\in (\frac{1}{8\pi y_3},\infty)$, and becomes an eigenvalue after crossing the threshold $\alpha=\frac{1}{8\pi y_3}$. At threshold, a zero energy resonance appears. In the Dirichlet case, a couple bound state-resonance separate at $z=0$ moving respectively on the positive imaginary axis (eigenvalue) and on the negative imaginary axis (anti-bound state); the couple bifurcates from the branch of resonances colliding at zero, which is an eigenvalue, as described in Proposition \ref{DirSpec}, iv). In both cases, with a different phenomenology, it should be noted that the eigenvalue persists (see for this terminology the recent paper \cite{CDG25}) and it does not disappear.
A general understanding of the phenomenon of persistence/disappearing of eigenvalues at threshold seems still lacking. 
\end{remark}

\begin{figure}
\centering
\includegraphics[width=0.8\textwidth]{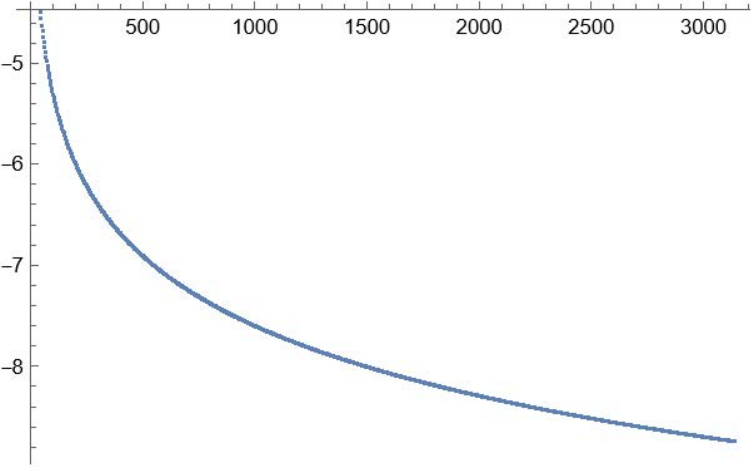}
\caption{Plot of the imaginary part versus real part of the first 100 values $z$ with $\Re{z}>\pi$ for which $z^{2}$ is a resonance for $\alpha=0$ and $y_{3}=1$ in the Dirichlet half-space case. }
\end{figure}

\section{Asymptotics of the Resonance Counting Function}
In the present section we study the asymptotic behaviour of the resonance counting function for the one center Laplacian in the half-space. Consider the set $\Sigma\Bigl(-\Delta_{\alpha, y}^{\textup{D}}\Bigr)$ of resonances of $-\Delta_{\alpha, y}^{\textup{D}}$. Let's define
\begin{equation*}
\mathcal{N}_{-\Delta_{\alpha, y}^{\textup{D}}}(R)=\#\Set{z\in\Sigma\Bigl(-\Delta_{\alpha, y}^{\textup{D}}\Bigr) |\abs{z}<R}
\end{equation*}
the counting function for the the elements of $\Sigma\Bigl(-\Delta_{\alpha, y}^{\textup{D}}\Bigr)$, counted with the appropriate multiplicity, contained in the disk of radius $R$ centered in zero. We now prove that (see \cite{LL:AsymptoticsResonancesInducedPointInteractions, Albeverio2020} for an analogous result for $N$-point interactions in $\R^{n}$)
\begin{proposition}
\label{propasymptotichalf-spaceDirichlet}
For all $\alpha\in\R$ and $y\in\mathbb{R}^{3}_+$ it holds that
\begin{equation}\label{CountingN}
\lim_{R\to +\infty}\frac{\mathcal{N}_{-\Delta_{\alpha, y}^{\textup{D}}}(R)}{R}=2 \frac{y_{3}}{\pi}.
\end{equation}
\end{proposition}
\begin{remark}
We actually count solutions of $\Gamma_{\alpha, y}^{\textup{D}}(z)=0$. So the result is an consequence of well known properties of distribution of zeroes of exponential polynomials (see for example Section 2 in  \cite{LL:AsymptoticsResonancesInducedPointInteractions} and references therein), but here we give an independent direct proof. 
\end{remark}
\begin{proof}
Let $z_{k}$ be the resonance of $-\Delta_{\alpha, y}^{\textup{D}}$ such that $\Re z_{k}:=a_k\in\Bigl(\frac{k\pi}{y_{3}}, \frac{(2k+1)\pi}{2y_{3}}\Bigr)$ and $\{a_{k}\}_{k=1, 2, \dots}$  the associated sequence. We start by proving that
\begin{equation}
\label{asymptoticroots}
\lim_{k\to +\infty}\frac{a_{k}}{2k+1/2}=\frac{\pi}{2y_{3}}.
\end{equation}
We observe that, even if it can occur that $a_{1}\in\Bigl(\frac{5\pi}{4y_{3}}, \frac{3\pi}{2y_{3}}\Bigr)$, it still happens that, definitely in $k$, $a_{k}\in\Bigl(\frac{k\pi}{y_{3}}, \frac{(2k+1/2)\pi}{2y_{3}}\Bigr)$. This is true because, recalling \cref{gandheqhalfspace}, we have that
\begin{equation*}
g\biggl(\frac{(2k+1/2)\pi}{2y_{3}}\biggr)=\frac{1}{(2k+1/2)\pi}\quad\textup{and}\quad \lim_{a\to\frac{(2k+1/2)\pi}{2y_{3}}}h(a)=e^{-8\pi\alpha y_{3}}
\end{equation*}
and so there exists a $\overline{k}$, such that for all $ k\ge\overline{k}$, $g<h$ when both evaluated at the middle point of the $k$-th interval. This fact together with \cref{ghconditions} implies that the unique solution is in the first half of the interval.\\ Let $w$ be in $\Bigl(\frac{\overline{k}\pi}{y_{3}}, \frac{(2\overline{k}+1/2)\pi}{2y_{3}}\Bigr)$ and $\{w_{l}\}_{l=0, 1,\dots}=\Bigl\{w+\frac{2l\pi}{2y_{3}}\Bigr\}_{l=0, 1, \dots}$. We have that
\begin{equation*}
g(w_{l})=\frac{\sin{(2y_{3}w)}}{2y_{3}w+2l\pi}\quad\textup{and}\quad h(w_{l})=e^{-8\pi\alpha y_{3}-\frac{2y_{3}w+2l\pi}{\tan{(2y_{3}w)}}}.
\end{equation*}
Both sequences are vanishing, but $\{h(w_{l})\}$ does it faster than $\{g(w_{l})\}$, which implies that for sufficiently large $l$, $g(w_{l})>h(w_{l})$ and so the solution in the $(\overline{k}+l)$-th interval must be in $\Bigl(w_{l}, \frac{(2(\overline{k}+l)+1/2)\pi}{2y_{3}}\Bigr)$. The arbitrarity of $w$ in $\Bigl(\frac{\overline{k}\pi}{y_{3}}, \frac{(2\overline{k}+1/2)\pi}{2y_{3}}\Bigr)$ leads to \cref{asymptoticroots}.

So we have that $a_{k}=\frac{(2k+1/2)\pi}{2y_{3}}+o(1)$ for $k\to+\infty$. (the symbol $o(1)$ here denotes a generic term vanishing as $k\to +\infty$. We use this property to estimate $\mathcal{N}_{-\Delta_{\alpha, y}^{\textup{D}}}(R)$ as $R\to+\infty$ excluding at most finitely many resonances. This is equivalent to find the highest $k$ for which
\begin{equation*}
a_{k}^{2}+\frac{1}{4y_{3}^{2}}\ln^{2}{\Biggl(\frac{\sin{(2y_{3}a_{k})}}{2y_{3}a_{k}}\Biggr)}<R^{2}.
\end{equation*}
Substituting the asymptotic expression for $k\to+\infty$, we get
\begin{equation*}
\biggl(2k+\frac{1}{2}\biggr)^{2}\pi^{2}+\ln^{2}{\Biggl(\biggl(2k+\frac{1}{2}\biggr)\pi\Biggr)}+o(1)<4y_{3}^{2}R^{2}.
\end{equation*}
If we set $s=\Bigl(2k+\frac{1}{2}\Bigr)\pi$, it becomes
\begin{equation*}
s^{2}+\ln^{2}{(s)}+o(1)<4y_{3}^{2}R^{2}.
\end{equation*}
We are concerned about large value of $s$, so, since the left side happens to be increasing for $s$ sufficiently large, the largest $s$ satisfying the inequality will be the one satisfying
\begin{equation*}
s^{2}+\ln^{2}{(s)}+o(1)=4y_{3}^{2}R^{2}.
\end{equation*}
We look for the asymptotic for $s$ solving the equation, when $R\to+\infty$. We write
\begin{equation}
\label{recursives}
s=\sqrt{4y_{3}^{2}R^{2}-\ln^{2}{(s)}}+o(1).
\end{equation}
We see that $s<2y_{3}R$ and $s=O(R)$. This implies that $\ln^{2}{(s)}=O\Bigl(\ln^{2}(R)\Bigr)$. Now we substitute this in \cref{recursives}, obtaining
\begin{equation*}
s=2y_{3}R\sqrt{1+O\Biggl(\frac{\ln^{2}{(R)}}{R^{2}}}\Biggr)+o(1)=2y_{3}R+o(1).
\end{equation*}
And so for $k$ holds ($\lfloor \cdot \rfloor$ and $\{\cdot\}$ indicate respectively the integer part and the fractional part)
$k=\Biggl\lfloor\frac{y_{3}}{\pi}R-\frac{1}{4}\Biggr\rfloor.$\\
But in this way we only counted the resonances with positive real part, so to count them all we just double the number:
\begin{equation}\label{CountingNfloor}
\mathcal{N}_{-\Delta_{\alpha, y}^{\textup{D}}}(R)=2\Biggl\lfloor\frac{y_{3}}{\pi}R-\frac{1}{4}\Biggr\rfloor+o(1)
\end{equation}
which implies equation \eqref{CountingN} in the statement
\end{proof}
\begin{remark} \label{asympdim}
From \eqref{CountingN} or from \eqref{CountingNfloor} it follows
\begin{equation*}
\begin{split}
\lim_{R\to+\infty}\frac{\ln{\Bigl(\mathcal{N}_{-\Delta_{\alpha, y}^{\textup{D}}}(R)\Bigr)}}{\ln{R}}=1
\end{split}
\end{equation*}
We give a comparison with analogous results for standard Schr\"odinger operators $-\Delta+V$. Firstly, for real potentials in $C^{\infty}_{\text{comp}}(\mathbb R^3)$ (see \cite{Melrosebook95}) or $L^{\infty}_{\text{comp}}(\mathbb R^3)\cap H^{\frac{n-3}{2}}$ (see \cite{SmithZworski16}) it is known that infinitely many resonances exist. Concerning asymptotic distribution of resonances, in \cite{CH2008}  (see also \cite{CH2005}) it is shown that for Schr\"odinger operators $-\Delta+V$ with generic $C^{\infty}_{\text{comp}}(\mathbb R^n)$ one has $$\limsup_{R\to+\infty}\frac{\ln{\Bigl(\mathcal{N}_{-\Delta +V}(R)\Bigr)}}{\ln{R}}=n$$ from which it follows the upper bound  $\mathcal{N}_{-\Delta +V} \leq C R^n$, where the optimal bound $\mathcal{N}_{-\Delta +V} = C R^n$ is attained for radial potentials (see \cite{Zworski89a}), but not in general. The asymptotic behavior was already noticed in \cite{LL:AsymptoticsResonancesInducedPointInteractions} for finitely many point interactions in $\R^3$, and compared with the analogous behavior on quantum graphs. See also the earlier work \cite{Zerzeri91}. While the asymptotic distribution given by \eqref{CountingN} seems somewhat anomalous in a three dimensional model, one has to take into account that point interactions can be recovered as limits of standard Schr\"odinger operators with rescaled smooth potentials, that however have to satisfy some very special conditions (see for example \cite{Albeverio:SolvableModelsQM}).
\end{remark}

\begin{remark}
The same result holds also for $\mathcal{N}_{-\Delta_{\alpha, y}^{\textup{N}}}(R)$. The proof for the Neumann case follows closely the one given above and it will be omitted.
\end{remark}

\section{Semiclassical Asymptotics of Resonances for the Half-Space}
In general, as also the previous treatment shows, it is difficult to obtain explicit information on all resonances. In order to understand some properties of them, one can study their asymptotic behavior on certain regimes. An example of these regimes is the high energy limit. To fix ideas we limit ourself to the single point interaction in the half-space. This limit consist in studying the distribution of the poles of
\begin{equation*}
\Bigl(-\Delta_{\alpha, y}^{{\textup{D}}}-z^{2}\Bigr)^{-1}, \quad \textup{for}\,\abs{z}\to +\infty
\end{equation*}
and the analogous distribution in the Neumann case $\Delta_{\alpha, y}^{{\textup{N}}}$.\
The high energy limit is an example of \emph{semiclassical limit}. This is obtained by rescaling $z$ trough a reference constant $h>0$ (who mimics Planck's constant) and passing to the limit $h\to 0$. In this way, the energy $\Bigl(\frac{z}{h}\Bigr)^{2}$ diverges, implying that the energies of the considered phenomena are well higher than the reference energy $h$. Correspondingly, we consider the behavior of the resolvent
\begin{equation*}
\Biggl(-\Delta_{\alpha, y}^{\textup D}-\biggl(\frac{z}{h}\biggr)^{2}\Biggr)^{-1}=h^{2}\Bigl(-h^{2}\Delta_{\alpha, y}^{\textup D}-z^{2}\Bigr)^{-1},\quad \textup{for}\, h\to 0.
\end{equation*}
In applications, typically, the potential which describes the interaction (here the point interaction) is usually assumed to be energy dependent and hence dependent on $h$ in this setting. A common choice for the interaction is of the form $h^{-\beta}$ with $\beta>0$ (see for example \cite{DM:SemiclassicalResonanceAsymptoticDeltaHalfLine} and \cite{Galkowski:ResonancesThinBarrierCircle}). We consider both positive and negative sign of the point charge. So the operator considered is $-h^{2}\Delta_{\pm h^{-\beta}, y}^{\textup D}$ and its resonances are the poles in the complex lower half-plane of
\begin{equation*}
\Biggl(-\Delta_{\pm h^{-\beta}, y}^{\textup D}-\biggl(\frac{z}{h}\biggr)^{2}\Biggr)^{-1}=h^{2}\Bigl(-h^{2}\Delta_{\pm h^{-\beta}, y}^{\textup D}-z^{2}\Bigr)^{-1},\quad \textup{for}\, h\to 0.
\end{equation*}
Below we will study the asymptotics of these resonances for $\Omega=\mathbb{R}^{3}_+$ with both Dirichlet and Neumann boundary conditions. The approach parallels the analysis in \cite{DM:SemiclassicalResonanceAsymptoticDeltaHalfLine}
where the one dimensional case of the half-line with a repulsive delta interaction at a point and Dirichlet boundary conditions at the origin is treated. This is sometimes called Winter's model in the literature (see \cite{Sacchetti23} for a recent analysis and references).\\ Here we study the half-space with both the signs of the delta interaction and we include the results for the Neumann b.c. at the end of the Section.
\subsection{Dirichlet Boundary Condition}
We start by studying the semiclassical resonance asymptotics for $-h^{2}\Delta_{\pm h^{-\beta}, y}^{\textup{D}}$. The form of the resolvent implies that the equation for the resonances is
\begin{equation}
\label{semiclassicalequation}
\pm h^{-\beta}-\frac{iz}{4\pi h}+\frac{e^{2iy_{3}\frac{z}{h}}}{8\pi y_{3}}=0,
\end{equation}
which can be rewritten as
\begin{equation*}
-e^{\pm8\pi y_{3}h^{-\beta}}=\biggl(-2iy_{3}\frac{z}{h}\pm8\pi y_{3}h^{-\beta}\biggr)e^{-2iy_{3}\frac{z}{h}\pm 8\pi y_{3}h^{-\beta}}.
\end{equation*}
Setting $-w^{\pm}=-e^{\pm8\pi y_{3}h^{-\beta}}$ and $x=-2iy_{3}\frac{z}{h}\pm8\pi y_{3}h^{-\beta}$, the previous equation becomes $-w=xe^{x}$. The main difference in the two cases is that in the first $w\to+\infty$ as $h\to 0^{+}$, while in the second $w\to 0$. 

The solutions $x$ in the complex plane of $-w=xe^{x}$ are a denumerable family denoted by $x=W_{k}(-w)$, where $W$ is the Lambert $W$ function (see \cite{CHK:LambertW}) and $k$ varies over $\Z$. Specifically we use the fact that $W_{k}$ can be expanded in a convergent series (note that this expansion is valid for both $w$ approaching zero and infinity)
\begin{equation*}
\begin{split}
W_{k}(-w) & = \ln{(-w)}+2\pi ik-\ln{\Bigl(\ln{(-w)}+2\pi i k\Bigr)}+R_{k}\\
               & = \ln{(w)}+(2k+1)i\pi-\ln{\Bigl(\ln{(w)}+(2k+1)i\pi\Bigr)}+R_{k},
\end{split}
\end{equation*}
where 
\begin{equation}
\label{seriesLambert}
R_{k}=\sum_{j=0}^{+\infty}\sum_{m=1}^{+\infty}c_{j, m}\frac{\ln^{m}{\Bigl(\ln{(w)}+(2k+1)i\pi\Bigr)}}{\Bigl(\ln{(w)}+(2k+1)i\pi\Bigr)^{j+m}}\ \ \ \ \text{and}\ \ \ \
c_{j, m}=\frac{(-1)^{j}}{m!}\stirlingone{j+m}{j+1}.
\end{equation}
Here $\stirlingone{p}{q}$ are the Stirling numbers, which are the number of ways to arrange $p$ objects in $q$ cycles.\\
The following Lemma holds.
\begin{lemma}
\label{lemmaboundresto}
The series \cref{seriesLambert} is absolutely convergent for $w$ large (or small) enough and $k\in\Z$. More precisely we have the tail estimate
\begin{equation}
\label{boundresto}
\begin{split}
\abs{R_{k}-\frac{\ln{\Bigl(\ln{(w)}+(2k+1)i\pi\Bigr)}}{\ln{(w)}+(2k+1)i\pi}} \le \sum_{j\ge 0, m\ge 1, (j, m)\ne (0, 1)}\abs{c_{j, m}\frac{\ln^{m}{\Bigl(\ln{(w)}+(2k+1)i\pi\Bigr)}}{\Bigl(\ln{(w)}+(2k+1)i\pi\Bigr)^{j+m}}}
                                                                                                            \le 2\abs{\frac{\ln{\Bigl(\ln{(w)}+(2k+1)i\pi\Bigr)}}{\ln{(w)}+(2k+1)i\pi}}^{2}
\end{split}
\end{equation}
and
\begin{equation}
\label{bounditerativelogarithm}
\abs{\frac{\ln{\Bigl(\ln{(w)}+(2k+1)i\pi\Bigr)}}{\ln{(w)}+(2k+1)i\pi}}\le\frac{1}{2}.
\end{equation}
\end{lemma}
\begin{proof}
We prove the Lemma for both $w$ large and small simultaneously. Inequality \cref{bounditerativelogarithm} is clearly fulfilled for $w$ sufficiently large (or small). To prove \cref{boundresto} we write
\begin{multline*}
\abs{R_{k}-\frac{\ln{\Bigl(\ln{(w)}+(2k+1)i\pi\Bigr)}}{\ln{(w)}+(2k+1)i\pi}}\le \sum_{j\ge 0, m\ge 1, (j, m)\ne (0, 1)}\abs{c_{j, m}\frac{\ln^{m}{\Bigl(\ln{(w)}+(2k+1)i\pi\Bigr)}}{\Bigl(\ln{(w)}+(2k+1)i\pi\Bigr)^{j+m}}}\\
= \sum_{m=2}^{+\infty}\abs{c_{0, m}\frac{\ln^{m}{\Bigl(\ln{(w)}+(2k+1)i\pi\Bigr)}}{\Bigl(\ln{(w)}+(2k+1)i\pi\Bigr)^{m}}}+\sum_{m=1}^{+\infty}\sum_{j=1}^{+\infty}\abs{c_{j, m}\frac{\ln^{m}{\Bigl(\ln{(w)}+(2k+1)i\pi\Bigr)}}{\Bigl(\ln{(w)}+(2k+1)i\pi\Bigr)^{j+m}}}=S_{1}+S_{2}.
\end{multline*}
We observe that $\stirlingone{m}{1}=(m-1)!$, and so $c_{0, m}=1/m$. Then
\begin{equation*}
S_{1}\le \frac{1}{2}\sum_{m=2}^{+\infty}\abs{\frac{\ln{\Bigl(\ln{(w)}+(2k+1)i\pi\Bigr)}}{\ln{(w)}+(2k+1)i\pi}}^{m}=\frac{1}{2}\sum_{m=2}^{+\infty}q^{m}= \frac{q^{2}}{2(1-q)}\le \abs{\frac{\ln{\Bigl(\ln{(w)}+(2k+1)i\pi\Bigr)}}{\ln{(w)}+(2k+1)i\pi}}^{2},
\end{equation*}
where the last inequality is valid because $q\le 1/2$ for $w$ large (or small) enough. We consider the other sum
\begin{equation}
\label{secondterm}
S_{2}=\frac{1}{\abs{\ln{(w)}+(2k+1)i\pi}}\sum_{m=1}^{+\infty}\abs{\frac{\ln{\Bigl(\ln{(w)}+(2k+1)i\pi\Bigr)}}{\ln{(w)}+(2k+1)i\pi}}^{m}\sum_{j=1}^{+\infty}\frac{\abs{c_{j, m}}}{\abs{\ln{(w)}+(2k+1)i\pi}^{j-1}}.
\end{equation}
Since the double sum is absolutely convergent for large or small $w$, then it must hold that
\begin{equation*}
\abs{\frac{\ln{\Bigl(\ln{(w)}+(2k+1)i\pi\Bigr)}}{\ln{(w)}+(2k+1)i\pi}}^{m}\sum_{j=1}^{+\infty}\frac{\abs{c_{j, m}}}{\abs{\ln{(w)}+(2k+1)i\pi}^{j-1}}\to0\quad\textup{for}\quad m\to +\infty.
\end{equation*}
Moreover, for large $w$, the first multiplicative term is decreasing in $w$ and each term in the sum is decreasing as well. Since all terms are positive, we can conclude that the expression above is decreasing for large $w$. This means that there is $N$ such that, setting $w=N$, for all $m\ge 1$ we have
\begin{equation*}
\abs{\frac{\ln{\Bigl(\ln{(N)}+(2k+1)i\pi\Bigr)}}{\ln{(N)}+(2k+1)i\pi}}^{m}\sum_{j=1}^{+\infty}\frac{\abs{c_{j, m}}}{\abs{\ln{(N)}+(2k+1)i\pi}^{j-1}}\le N.
\end{equation*}
In the same way we have that, for small $w$, the first multiplicative term is increasing in $w$ and each term in the sum is increasing and positive as well. So the whole expression is increasing for small $w$ and that means that ther is $N$ such that if $w=\frac{1}{N}$, for all $m\ge 1$
\begin{equation*}
\abs{\frac{\ln{\Bigl(-\ln{(N)}+(2k+1)i\pi\Bigr)}}{-\ln{(N)}+(2k+1)i\pi}}^{m}\sum_{j=1}^{+\infty}\frac{\abs{c_{j, m}}}{\abs{-\ln{(N)}+(2k+1)i\pi}^{j-1}}\le N
\end{equation*}
holds. To consider together both cases, we can just put $\sgn{\Bigl(\ln{(w)}\Bigr)}$ in front of $\ln{(N)}$.
Then, for $\abs{w}\ge N$ large (or $\abs{w}\le \frac{1}{N}$ respectively), it holds that
\begin{equation*}
\begin{split}
\sum_{j=1}^{+\infty}\frac{\abs{c_{j, m}}}{\abs{\ln{(w)}+(2k+1)i\pi}^{j-1}} \le \sum_{j=1}^{+\infty}\frac{\abs{c_{j, m}}}{\abs{\sgn{\Bigl(\ln{(w)}\Bigr)}\ln{(N)}+(2k+1)i\pi}^{j-1}}
                                                                                                               \le N\abs{\frac{\sgn{\Bigl(\ln{(w)}\Bigr)}\ln{(N)}+(2k+1)i\pi}{\ln{\Bigl(\sgn{\Bigl(\ln{(w)}\Bigr)}\ln{(N)}+(2k+1)i\pi\Bigr)}}}^{m}.
\end{split}
\end{equation*}
Substituting this in \cref{secondterm}, we get
\begin{equation*}
\begin{split}
S_{2} & \le \frac{N}{\abs{\ln{(w)}+(2k+1)i\pi}}\sum_{m=1}^{+\infty}\abs{\frac{\ln{\Bigl(\ln{(w)}+(2k+1)i\pi\Bigr)}}{\ln{(w)}+(2k+1)i\pi}\cdot\frac{\sgn{\Bigl(\ln{(w)}\Bigr)}\ln{(N)}+(2k+1)i\pi}{\ln{\Bigl(\sgn{\Bigl(\ln{(w)}\Bigr)}\ln{(N)}+(2k+1)i\pi\Bigr)}}}^{m}\\
         & = 2N\frac{\abs{\ln{\Bigl(\ln{(w)}+(2k+1)i\pi\Bigr)}}}{\abs{\ln{(w)}+(2k+1)i\pi}^{2}}\abs{\frac{\sgn{\Bigl(\ln{(w)}\Bigr)}\ln{(N)}+(2k+1)i\pi}{\ln{\Bigl(\sgn{\Bigl(\ln{(w)}\Bigr)}\ln{(N)}+(2k+1)i\pi\Bigr)}}}\frac{1}{2}\sum_{m=0}^{+\infty}p^{m},
\end{split}
\end{equation*}
where
\begin{equation*}
\begin{split}
p = \abs{\frac{\ln{\Bigl(\ln{(w)}+(2k+1)i\pi\Bigr)}}{\ln{(w)}+(2k+1)i\pi}\cdot\frac{\sgn{\Bigl(\ln{(w)}\Bigr)}\ln{(N)}+(2k+1)i\pi}{\ln{\Bigl(\sgn{\Bigl(\ln{(w)}\Bigr)}\ln{(N)}+(2k+1)i\pi\Bigr)}}}
  =
\begin{cases}
\frac{s(w)}{s(N)}\quad w\,\textup{large}\\
\frac{s(w)}{s\Bigl(\frac{1}{N}\Bigr)}\quad w\,\textup{small}
\end{cases}.
\end{split}
\end{equation*}
It also holds that $\lim_{w\to+\infty}s(w)=0$, so for large $w$, $p<\frac{1}{2}$. For small $w$ instead $s$ is increasing and $\lim_{w\to 0}s(w)=0$, which means that for sufficiently small $w$, $p<\frac{1}{2}$. But,
$\frac{1}{2}\sum_{m=0}^{+\infty}p^{m}=\frac{1}{2(1-p)}<1$
for $p<\frac{1}{2}$. This implies
\begin{equation*}
S_{2}\le\frac{2N\abs{\sgn{\Bigl(\ln{(w)}\Bigr)}\ln{(N)}+(2k+1)i\pi\Bigr)}}{\abs{\ln{\Bigl(\sgn{\Bigl(\ln{(w)}\Bigr)}\ln{(N)}+(2k+1)i\pi\Bigr)}}}\cdot\frac{\abs{\ln{\Bigl(\ln{(w)}+(2k+1)i\pi\Bigr)}}}{\abs{\ln{(w)}+(2k+1)i\pi}^{2}}.
\end{equation*}
Given that $\abs{\ln{\Bigl(\ln{(w)}+(2k+1)i\pi\Bigr)}}\to+\infty$ for both $w\to+\infty$ and $w\to 0$, we have that
\begin{equation*}
\frac{2N\abs{\sgn{\Bigl(\ln{(w)}\Bigr)}\ln{(N)}+(2k+1)i\pi\Bigr)}}{\abs{\ln{\Bigl(\sgn{\Bigl(\ln{(w)}\Bigr)}\ln{(N)}+(2k+1)i\pi\Bigr)}}}\le\abs{\ln{\Bigl(\ln{(w)}+(2k+1)i\pi\Bigr)}},
\end{equation*}
for $w$ large (respectively small) enough. So one concludes that
$S_{2}\le\abs{\frac{\ln{\Bigl(\ln{(w)}+(2k+1)i\pi\Bigr)}}{\ln{(w)}+(2k+1)i\pi}}^{2}$
and the thesis follows.
\end{proof}

So we can write the resonances $z_{k}^{\pm}$ in terms of the Lambert $W$ function. For $-h^{2}\Delta_{h^{-\beta}, y}^{\textup{D}}$ they are given by
\begin{equation}
\label{zkexpression}
\begin{split}
z_{k}^{+} &= \frac{ih}{y_3}\biggl(W_{k}(-w)-8\pi y_{3}h^{-\beta}\biggr)\\
        & = \frac{ih}{y_3}\biggl(\ln{(w)}+(2k+1)i\pi-\ln{\Bigl(\ln{(w)}+(2k+1)i\pi\Bigr)}-8\pi y_3h^{-\beta}+R_{k}\biggr)\\
        & = \frac{ih}{y_3}\biggl((2k+1)i\pi-\ln{\Bigl(\ln{(w)}+(2k+1)i\pi\Bigr)}+R_{k}\biggr),
\end{split}
\end{equation}
where $k$ varies over $\Z$ and $w=e^{8\pi y_{3}h^{-\beta}}$. Instead, for $-h^{2}\Delta_{-h^{-\beta}, y}^{\textup{D}}$ (here $w=e^{-8\pi y_{3}h^{-\beta}}$)
\begin{equation}
\label{zkexpressionnegative}
\begin{split}
z_{k}^{-} &= \frac{ih}{y_3}\biggl(W_{k}(-w)+8\pi y_{3}h^{-\beta}\biggr)\\
        & = \frac{ih}{y_3}\biggl(\ln{(w)}+(2k+1)i\pi-\ln{\Bigl(\ln{(w)}+(2k+1)i\pi\Bigr)}+8\pi y_3h^{-\beta}+R_{k}\biggr)\\
        & = \frac{ih}{y_3}\biggl((2k+1)i\pi-\ln{\Bigl(\ln{(w)}+(2k+1)i\pi\Bigr)}+R_{k}\biggr)\ .
\end{split}
\end{equation}

Now we show that under some assumptions over $z_{k}$, $k$ is roughly of size $h^{-1}$.
\begin{lemma}
\label{lemmainequalityk}
Let $\varepsilon\in (0, 1)$ be given. Then, for $k$ such that $z_{k}$ is given by \cref{zkexpression} or \cref{zkexpressionnegative} and 
\begin{equation}
\label{inequalityzk}
\varepsilon\le\abs{z_{k}}\le1/\varepsilon,
\end{equation}
we have
\begin{equation}
\label{inequalityk}
\frac{\varepsilon}{2}\le\frac{\abs{k}\pi h}{y_{3}}\le\frac{2}{\varepsilon}.
\end{equation}
\end{lemma}
\begin{proof}
Using \cref{bounditerativelogarithm} and the inverse triangular inequality we have that
\begin{equation*}
\abs{R_{k}-\frac{\ln{\Bigl(\ln{(w)}+(2k+1)i\pi\Bigr)}}{\ln{(w)}+(2k+1)i\pi}}\ge\abs{\abs{R_{k}}-\frac{1}{2}}.
\end{equation*}
Combining this with \cref{boundresto} and \cref{bounditerativelogarithm} we get
$\abs{\abs{R_{k}}-\frac{1}{2}}\le\frac{1}{2}$,
equivalent to $\abs{R_{k}}\le 1$. We now prove the first inequality in \cref{inequalityk}. For the sake of contradiction we suppose that $\frac{\abs{k}\pi h}{y_{3}}<\frac{\varepsilon}{2}$. Combining the first inequality of \cref{inequalityzk}, $\abs{R_{k}}\le 1$ and \cref{zkexpression}, we get
\begin{equation*}
\varepsilon\le\abs{z_{k}}\le\frac{\abs{k}\pi h}{y_{3}}+\frac{h}{y_3}\abs{\ln{\Bigl(\ln{(w)}+(2k+1)i\pi\Bigr)}}+\frac{\pi+1}{2y_{3}}h.
\end{equation*}
Then, using $\frac{\abs{k}\pi h}{y_{3}}<\frac{\varepsilon}{2}$, we have
$\varepsilon<\frac{h}{y_{3}}\abs{\ln{\Bigl(\ln{(w)}+(2k+1)i\pi\Bigr)}}+\frac{\pi+1}{y_{3}}h$.
Since the second term in the right side goes to zero faster than the first, it must hold that, for $h$ small enough
$\frac{\varepsilon}{2}<\frac{h}{y_{3}}\abs{\ln{\Bigl(\ln{(w)}+(2k+1)i\pi\Bigr)}}$,
which implies
\begin{equation*}
e^{\frac{\varepsilon y_{3}}{2h}}\le\abs{\ln{(w)}+(2k+1)i\pi}\le8\pi y_{3}h^{-\beta}+2\abs{k}\pi+\pi\le8\pi y_{3}h^{-\beta}+\varepsilon y_{3}h^{-1}
\end{equation*}
and this is absurd for $h$ small enough.

To prove the second inequality \cref{inequalityk} we suppose for the sake of contradiction that $\frac{\abs{k}\pi h}{y_{3}}>\frac{2}{\varepsilon}$. By $\abs{R_{k}}\le 1$, \cref{zkexpression} and triangular inequality we have that
\begin{equation*}
\abs{z_{k}}\ge\frac{\abs{k}\pi h}{y_{3}}-\frac{h}{2y_3}\abs{\ln{\abs{\ln{(w)}+(2k+1)i\pi}}}-\frac{h}{2y_{3}}\abs{\Arg{\Bigl(\ln{(w)}+(2k+1)i\pi\Bigr)}}-\frac{\pi+1}{2y_{3}}h.
\end{equation*}
This, combined with the second inequality of \cref{inequalityzk}, implies
\begin{equation*}
\begin{split}
\frac{\abs{k}\pi h}{y_{3}}-\frac{h}{2y_3}\abs{\ln{\abs{\ln{(w)}+(2k+1)i\pi}}} \le \frac{1}{\varepsilon}+\frac{h}{2y_{3}}\abs{\Arg{\Bigl(\ln{(w)}+(2k+1)i\pi\Bigr)}}+\frac{\pi+1}{2y_{3}}h
                                                                                                                                        \le \frac{3}{2\varepsilon},
\end{split}
\end{equation*}
the last inequality being valid for $h$ small enough. So it holds that
$2\abs{k}\pi-\frac{3y_{3}}{\varepsilon h}\le\abs{\ln{\abs{\ln{(w)}+(2k+1)i\pi}}}$ and hence
$e^{2\abs{k}\pi-\frac{3y_{3}}{\varepsilon h}}-2\abs{k}\pi\le\abs{\ln{(w)}}+\pi$.
But, since the function $f(x)=e^{x-a}-x$ is increasing in $(a, +\infty)$ and $2\abs{k}\pi>\frac{4y_{3}}{\varepsilon h}$ we have $f\Bigl(\frac{4y_{3}}{\varepsilon h}\Bigr)<f\Bigl(2\abs{k}\pi\Bigr)\le\ln{(w)}+\pi$ and finally
\begin{equation*}
e^{\frac{y_{3}}{\varepsilon h}}\le8\pi y_{3}h^{-\beta}+\frac{4y_{3}}{\varepsilon h}+\pi,
\end{equation*}
which is a contradiction for $h$ small enough.
\end{proof}
We now establish asymptotic behaviour for resonances apart from the ones near zero and near infinity. We distinguish between $0<\beta<1$ and $\beta>1$. This different treatment is needed because different values of $\beta$ changes the relative dominance between $\ln{(y)}$ and $\abs{2k+1}\pi$ in \cref{zkexpression}. We start with the case $\beta\in(0, 1)$.
\begin{proposition}
Let $\beta\in(0, 1)$ and $\varepsilon\in(0, 1)$ be given. Then there is $h_{0}>0$ such that, when $h\in(0, h_{0}]$, all solutions of \cref{semiclassicalequation} satisfying
$\varepsilon\le\abs{z}\le{1}/{\varepsilon}$,
obey
\begin{equation*}
0\le-\Im{z}-\frac{h}{2y_{3}}\ln{\Bigl(2y_{3}h^{-1}\abs{\Re{z}}\Bigr)}\le\frac{72\pi^{2}}{y_{3}}\varepsilon^{-2}h^{3-2\beta}.
\end{equation*}
\end{proposition}
\begin{proof}
Since $\beta\in(0, 1)$, then the first inequality in \cref{inequalityk} implies
\begin{equation}
\label{dominanttermbeta<1}
\abs{(2k+1)i\pi}\ge 2\abs{k}\pi-\pi\ge\varepsilon y_{3}h^{-1}-\pi>\abs{\ln{(w)}}=8\pi y_{3}h^{-\beta},
\end{equation}
for $h$ small enough, so that it is more convenient to rewrite
\begin{equation*}
\ln{\Bigl(\ln{(w)}+(2k+1)i\pi\Bigr)}=\ln{\Bigl((2k+1)i\pi\Bigr)}+\ln{\biggl(1+\frac{\ln{(w)}}{(2k+1)i\pi}\biggr)}
\end{equation*}
and
\begin{equation*}
z_{k}^{\pm}=\frac{ih}{2y_{3}}\biggl((2k+1)i\pi-\ln{\Bigl((2k+1)i\pi\Bigr)}+R'_{k}\biggr), \quad\textup{where}\quad R'_{k}=R_{k}-\ln{\biggl(1+\frac{\ln{(w)}}{(2k+1)i\pi}\biggr)}.
\end{equation*}
Using $\abs{R_{k}}\le 1$ and \cref{dominanttermbeta<1}, we get
\begin{equation}
\label{boundRk'}
\abs{R'_{k}}\le\abs{R_{k}}+\abs{\ln{\abs{1+\frac{\ln{(w)}}{(2k+1)i\pi}}}}+\abs{\Arg{\biggl(1+\frac{\ln{(w)}}{(2k+1)i\pi}\biggr)}}\le 1+\ln{2}+\frac{\pi}{2}\le 4.
\end{equation}
Hence, by taking the real and imaginary part of $z_{k}^{\pm}$ we get
\begin{equation}
\label{realimaginaryzk}
\begin{cases}
\Re{z_{k}^{\pm}}=-\frac{(2k+1)\pi h}{2y_{3}}+\frac{h}{2y_{3}}\Arg{\Bigl((2k+1)i\pi\Bigr)}-\frac{h}{2y_{3}}\Im{R'_{k}}\\
\Im{z_{k}^{\pm}}=-\frac{h}{2y_{3}}\ln{\abs{(2k+1)\pi}}+\frac{h}{2y_{3}}\Re{R'_{k}}
\end{cases}.
\end{equation}
We now separate real and imaginary part of \cref{semiclassicalequation} 
\begin{equation*}
\begin{cases}
e^{-2\frac{y_3}{h}\Im{z_{k}^{\pm}}}\cos{\biggl(2\frac{y_{3}}{h}\Re{z_{k}^{\pm}}\biggr)}=-2\frac{y_3}{h}\Im{z_{k}^{\pm}}\mp8\pi y_{3}h^{-\beta}\\
e^{-2\frac{y_3}{h}\Im{z_{k}^{\pm}}}\sin{\biggl(2\frac{y_{3}}{h}\Re{z_{k}^{\pm}}\biggr)}=2\frac{y_3}{h}\Re{z_{k}^{\pm}}
\end{cases}.
\end{equation*}
Squaring the equations and adding them side by side we have
\begin{equation*}
e^{-4\frac{y_3}{h}\Im{z_{k}^{\pm}}}=\frac{4y_{3}^{2}}{h^{2}}\biggl(\Bigl(\Im{z_{k}^{\pm}}\Bigr)^{2}+\Bigl(\Re{z_{k}^{\pm}}\Bigr)^{2}\biggr)\pm32\pi y_{3}^{2}h^{-\beta-1}\Im{z_{k}^{\pm}}+64\pi^{2}y_{3}^{2}h^{-2\beta}
\end{equation*}
or equivalently
\begin{equation*}
-4\frac{y_3}{h}\Im{z_{k}^{\pm}}=2\ln{\Bigl(2y_{3}h^{-1}\abs{\Re{z_{k}^{\pm}}}\Bigr)}+\ln{(1+t)} \ \ \text{where}\ \ \ 
t=\frac{4y_{3}^{2}\Bigl(\Im{z_{k}^{\pm}}\Bigr)^{2}\pm32\pi y_{3}^{2}h^{1-\beta}\Im{z_{k}^{\pm}}+64\pi^{2}y_{3}^{2}h^{2-2\beta}}{4y_{3}^{2}\Bigl(\Re{z_{k}^{\pm}}\Bigr)^{2}}.
\end{equation*}
We recall that $0\le\ln{(1+t)}\le t$ to get
\begin{equation}
\label{lastpassage}
0\le-4\frac{y_3}{h}\Im{z_{k}^{\pm}}-2\ln{\Bigl(2y_{3}h^{-1}\abs{\Re{z_{k}^{\pm}}}\Bigr)}\le t.
\end{equation}
From \cref{realimaginaryzk} we have, using \cref{inequalityk}, \cref{boundRk'} and $\Arg{\Bigl((2k+1)i\pi\Bigr)}=\pi/2$, that, for sufficiently small $h$
\begin{equation*}
\begin{split}
\abs{\Re{z_{k}}}& \ge \frac{\abs{2k+1}\pi h}{2y_{3}}-\abs{\frac{h}{2y_{3}}\Im{R'_{k}}-\frac{h}{2y_{3}}\Arg{\Bigl((2k+1)i\pi\Bigr)}}\\
                         & \ge \frac{\abs{k}\pi h}{y_{3}}-\frac{\pi h}{2y_{3}}-\frac{h}{2y_{3}}\abs{\Im{R'_{k}}}-\frac{h}{2y_{3}}\abs{\Arg{\Bigl((2k+1)i\pi\Bigr)}}\\
                         & \ge \frac{\varepsilon}{2}-\frac{\pi h}{2y_{3}}-\frac{2h}{y_{3}}-\frac{\pi h}{4y_{3}}\ge\frac{\varepsilon}{3}
\end{split}
\end{equation*}
and
\begin{equation*}
\begin{split}
\abs{\Im{z_{k}^{\pm}}}& \le \frac{h}{2y_{3}}\ln{\Bigl(2\abs{k}\pi+\pi\Bigr)}+\frac{h}{2y_{3}}\abs{\Re{R'_{k}}}\\
                         & = \frac{h}{2y_{3}}\ln{\Bigl(2\abs{k}\pi\Bigr)}+\frac{h}{2y_{3}}\ln{\biggl(1+\frac{1}{2\abs{k}\pi}\biggr)}+\frac{h}{2y_{3}}\abs{\Re{R'_{k}}}\\
                        & \le \frac{h}{2y_{3}}\ln{\biggl(\frac{4y_{3}}{\varepsilon h}\biggr)}+\frac{\ln{2}+2}{y_{3}}h\le\frac{h}{y_{3}}\ln{\biggl(\frac{4y_{3}}{\varepsilon h}\biggr)}.
\end{split}
\end{equation*}
For $z_{k}^{+}$ this, together with $\Im{z_{k}^{+}}<0$, implies
\begin{equation*}
t\le\frac{4h^{2}\ln^{2}{\Bigl(\frac{4y_{3}}{\varepsilon h}\Bigr)}+64\pi^{2}y_{3}^{2}h^{2-2\beta}}{4y_{3}^{2}\Bigl(\Re{z_{k}^{+}}\Bigr)^{2}}\le\frac{32\pi^{2}h^{2-2\beta}}{\Bigl(\Re{z_{k}^{+}}\Bigr)^{2}}\le 288\pi^{2}\varepsilon^{-2}h^{2-2\beta},
\end{equation*}
while for $z_{k}^{-}$
\begin{equation*}
t\le\frac{4h^{2}\ln^{2}{\Bigl(\frac{4y_{3}}{\varepsilon h}\Bigr)}+32\pi y_3h^{2-\beta}\ln{\Bigl(\frac{4y_{3}}{\varepsilon h}\Bigr)}+64\pi^{2}y_{3}^{2}h^{2-2\beta}}{4y_{3}^{2}\Bigl(\Re{z_{k}^{-}}\Bigr)^{2}}\le\frac{32\pi^{2}h^{2-2\beta}}{\Bigl(\Re{z_{k}^{-}}\Bigr)^{2}}\le 288\pi^{2}\varepsilon^{-2}h^{2-2\beta},
\end{equation*}
Plugging this inequality into \cref{lastpassage} gives the desired inequality.
\end{proof}
A different behavior happens for $\beta>1$.
\begin{proposition}
\label{propositionbeta>1}
Let $\beta>1$ and $\varepsilon\in (0, 1)$ be given. Then there is $h_{0}>0$ such that, when $h\in(0, h_{0}]$, all solutions to \cref{semiclassicalequation}$^{+}$ satisfying
$\varepsilon\le\abs{z}\le{1}/{\varepsilon}$
obey
\begin{equation*}
\abs{\Im{z}+\frac{2y_{3}\ln{\Bigl(8\pi y_{3}h^{-\beta}\Bigr)}}{h(2k+1)^{2}\pi^{2}}\Bigl(\Re{z}\Bigr)^{2}}\le\frac{1+96\varepsilon^{-4}}{4\pi y_{3}^{2}}h^{\beta+1}\ln{\Bigl(8\pi y_{3}h^{-\beta}\Bigr)}+\frac{\varepsilon^{-2}h^{2\beta-1}}{2\pi^{2}y_{3}},
\end{equation*}
which, using \cref{inequalityk}, implies
\begin{gather}
\label{upperboundbeta>1}
\Im{z}+\frac{\varepsilon^{2}}{32y_{3}}h\ln{\Bigl(8\pi y_{3}h^{-\beta}\Bigr)}\Bigl(\Re{z}\Bigr)^{2}\le\frac{1+96\varepsilon^{-4}}{4\pi y_{3}^{2}}h^{\beta+1}\ln{\Bigl(8\pi y_{3}h^{-\beta}\Bigr)}+\frac{\varepsilon^{-2}h^{2\beta-1}}{2\pi^{2}y_{3}}\\
\label{lowerboundbeta>1}
\Im{z}+\frac{8}{\varepsilon^{2}y_{3}}h\ln{\Bigl(8\pi y_{3}h^{-\beta}\Bigr)}\Bigl(\Re{z}\Bigr)^{2}\ge-\frac{1+96\varepsilon^{-4}}{4\pi y_{3}^{2}}h^{\beta+1}\ln{\Bigl(8\pi y_{3}h^{-\beta}\Bigr)}-\frac{\varepsilon^{-2}h^{2\beta-1}}{2\pi^{2}y_{3}}.
\end{gather}
\end{proposition}
\begin{proof}
Since $\beta>1$, by the second inequality in \cref{inequalityk}, we have
\begin{equation*}
\abs{2k+1}\pi\le 2\abs{k}\pi+\pi\le4y_{3}\varepsilon^{-1}h^{-1}+\pi<8\pi y_{3}h^{-\beta}=\ln{(w)},
\end{equation*}
for $h$ small enough. Recalling \cref{zkexpression}, we write
\begin{equation*}
z_{k}=\frac{ih}{2y_{3}}\Biggl((2k+1)i\pi-\ln{\Bigl(8\pi y_{3}h^{-\beta}\Bigr)}-\ln{\biggl(1+\frac{i(2k+1)}{8y_{3}}h^{\beta}\biggr)}+R_{k}\Biggr)
\end{equation*}
and so we can ensure that the third term is small for $h$ small. Separating both the real and imaginary part gives
\begin{equation}
\label{realimaginarypartzkbeta>1}
\begin{cases}
\Re{z_{k}}=\frac{h}{2y_{3}}\Biggr(-(2k+1)\pi+\Arg{\biggr(1+\frac{i(2k+1)}{8y_{3}}h^{\beta}\biggr)}-\Im{R_{k}}\Biggr)\\
\Im{z_{k}}=\frac{h}{2y_{3}}\Biggr(-\ln{\Bigl(8\pi y_{3}h^{-\beta}\Bigr)}-\ln{\abs{1+\frac{i(2k+1)}{8y_{3}}h^{\beta}}}+\Re{R_{k}}\Biggr)
\end{cases}.
\end{equation}
We observe that
\begin{equation*}
\ln{\abs{1+\frac{i(2k+1)}{8y_{3}}h^{\beta}}}=\frac{1}{2}\ln\biggl(1+\frac{(2k+1)^{2}}{64y_{3}^{2}}h^{2\beta}\biggr)=\frac{1}{2}\ln{(1+t)},
\end{equation*}
with $t=\frac{(2k+1)^{2}}{64y_{3}^{2}}h^{2\beta}$. 
Considering the imaginary part in \cref{realimaginarypartzkbeta>1}, we get
\begin{equation}
\label{absImzk}
\abs{\Im{z_{k}}+\frac{h}{2y_{3}}\ln{\Bigl(8\pi y_{3}h^{-\beta}\Bigr)}}=\frac{h}{2y_{3}}\abs{\Re{R_{k}}-\ln{\abs{1+\frac{i(2k+1)}{8y_{3}}h^{\beta}}}}.
\end{equation}

To deal with $R_{k}$, we recall that $w=8\pi y_{3}h^{-\beta}$ and consider
\begin{equation*}
\begin{split}
\abs{\frac{\ln{\Bigl(\ln{(w)}+(2k+1)i\pi\Bigr)}}{\ln{(w)}+(2k+1)i\pi}} \le \frac{\abs{\ln{\abs{\ln{(w)}+(2k+1)i\pi}}+i\frac{\pi}{2}}}{\abs{8\pi y_{3}h^{-\beta}+(2k+1)i\pi}}
                                                                                                    =\frac{\ln{\Bigl(h^{-\beta}\Bigr)}}{h^{-\beta}}\frac{\abs{\ln{\abs{\ln{(w)}+(2k+1)i\pi}}+i\frac{\pi}{2}}}{\ln{\Bigl(h^{-\beta}\Bigr)}\abs{8\pi y_{3}+(2k+1)i\pi h^{\beta}}},
\end{split}
\end{equation*}
but
\begin{equation*}
\frac{\abs{\ln{\abs{\ln{(w)}+(2k+1)i\pi}}+i\frac{\pi}{2}}}{\ln{\Bigl(h^{-\beta}\Bigr)}\abs{8\pi y_{3}+(2k+1)i\pi h^{\beta}}}\le\frac{1}{h^{\beta}\ln{(h^{-\beta})}}\Biggl(\abs{\frac{\ln{\abs{\ln{(w)}+(2k+1)i\pi}}}{\ln{(w)}+(2k+1)\pi}}+\frac{\frac{\pi}{2}}{\abs{{\ln{(w)}+(2k+1)\pi}}}\Biggr).
\end{equation*}
Since
\begin{equation*}
\lim_{h\to 0^{+}}\frac{1}{h^{\beta}\ln{(h^{-\beta})}}\abs{\frac{\ln{\abs{\ln{(w)}+(2k+1)i\pi}}}{\ln{(w)}+(2k+1)\pi}}=\frac{1}{8\pi y_{3}}
\end{equation*}
\begin{equation*}
\lim_{h\to 0^{+}}\frac{1}{h^{\beta}\ln{(h^{-\beta})}}\frac{\frac{\pi}{2}}{\abs{{\ln{(w)}+(2k+1)\pi}}}=0\ ,
\end{equation*}
then
\begin{equation*}
\frac{\abs{\ln{\abs{\ln{(w)}+(2k+1)i\pi}}+i\frac{\pi}{2}}}{\ln{\Bigl(h^{-\beta}\Bigr)}\abs{8\pi y_{3}+(2k+1)i\pi h^{\beta}}}\le\frac{1}{6\pi y_{3}},
\end{equation*}
for $h$ small enough and finally this means that
\begin{equation}
\label{bounditerativelogarithmbeta>1}
\abs{\frac{\ln{\Bigl(\ln{(w)}+(2k+1)i\pi\Bigr)}}{\ln{(w)}+(2k+1)i\pi}}\le\frac{h^{\beta}\ln{\Bigl(h^{-\beta}\Bigr)}}{4\pi y_{3}}.
\end{equation}
By \cref{boundresto}, we can also write
\begin{equation}
\label{boundrestobeta>1}
\abs{R_{k}-\frac{\ln{\Bigl(\ln{(w)}+(2k+1)i\pi\Bigr)}}{\ln{(w)}+(2k+1)i\pi}}\le\frac{h^{2\beta}\ln^{2}{\Bigl(h^{-\beta}\Bigr)}}{8\pi^{2}y_{3}^{2}}.
\end{equation}

Turning back to \cref{absImzk}, we now have
\begin{multline*}
\abs{\Im{z_{k}}+\frac{h}{2y_{3}}\ln{\Bigl(8\pi y_{3}h^{-\beta}\Bigr)}}\le \frac{h}{2y_{3}}\abs{\Re{R_{k}}-\Re{\Biggl(\frac{\ln{\Bigl(\ln{(w)}+(2k+1)i\pi\Bigr)}}{\ln{(w)}+(2k+1)i\pi}\Biggr)}}\\
+\frac{h}{2y_{3}}\abs{\Re{\Biggl(\frac{\ln{\Bigl(\ln{(w)}+(2k+1)i\pi\Bigr)}}{\ln{(w)}+(2k+1)i\pi}\Biggr)}-\ln{\abs{1+\frac{i(2k+1)}{8y_{3}}h^{\beta}}}}\le\frac{h^{2\beta+1}\ln^{2}{\Bigl(h^{-\beta}\Bigr)}}{16\pi^{2}y_{3}^{3}}+\abs{A}.
\end{multline*}
We now find a bound for $\abs{A}$. After some manipulations we can write
\begin{equation*}
{A=\frac{\biggl(\ln{\Bigl(8\pi y_{3}h^{-\beta}\Bigr)}+\frac{1}{2}\ln{(1+t)}\biggr)\frac{h^{\beta+1}}{16\pi y_{3}^{2}}+\frac{(2k+1)h^{2\beta+1}}{128\pi y_{3}^{3}}\Arg{\biggl(8\pi y_{3}h^{-\beta}+(2k+1)i\pi\biggr)}}{1+t}-\frac{h}{4y_{3}}\ln{(1+t)}}.
\end{equation*}
By triangle inequality and $t>0$ follows that
\begin{align*}
\abs{A}&\le\frac{h^{\beta+1}}{16\pi y_{3}^{2}}\ln{\Bigl(8\pi y_{3}h^{-\beta}\Bigr)}+\frac{\abs{2k+1}h^{2\beta+1}}{128\pi y_{3}^{3}}\abs{\Arg{\biggl(8\pi y_{3}h^{-\beta}+(2k+1)i\pi\biggr)}}\\
&+\frac{h}{4y_{3}}\ln{(1+t)}+\frac{h^{\beta+1}}{32\pi y_{3}^{2}}\ln{(1+t)}.
\end{align*} 
This expression can be also bounded using $\abs{\Arg{\biggl(8\pi y_{3}h^{-\beta}+(2k+1)i\pi\biggr)}}\le\frac{\pi}{2}$, $\ln{(1+t)}\le t$, the expression for $t$, $(2k+1)^{2}\le2(4k^{2}+1)$ and the second inequality in \cref{inequalityk} as
\begin{equation*}
\abs{A}\le\frac{h^{\beta+1}}{16\pi y_{3}^{2}}\ln{\Bigl(8\pi y_{3}h^{-\beta}\Bigr)}+\frac{h^{3\beta+1}}{1024\pi y_3^{4}}+\frac{\varepsilon^{-2}h^{3\beta-1}}{64\pi^{3}y_{3}^{2}}+\frac{3h^{2\beta+1}}{256y_{3}^{3}}+\frac{\varepsilon^{-1}h^{2\beta}}{64\pi y_{3}^{2}}+\frac{\varepsilon^{-2}h^{2\beta-1}}{8\pi^{2}y_{3}}.
\end{equation*}
For $\beta>1$, the dominant terms for small $h$ are the first and the last term. More specifically the first is dominant for $\beta\ge 2$ and the other for $1<\beta<2$. So, for $h$ small enough
\begin{equation*}
\abs{A}\le\frac{h^{\beta+1}}{8\pi y_{3}^{2}}\ln{\Bigl(8\pi y_{3}h^{-\beta}\Bigr)}+\frac{\varepsilon^{-2}h^{2\beta-1}}{4\pi^{2}y_{3}},
\end{equation*}
which implies
\begin{equation*}
\begin{split}
\abs{\Im{z_{k}}+\frac{h}{2y_{3}}\ln{\Bigl(8\pi y_{3}h^{-\beta}\Bigr)}} \le \frac{h^{2\beta+1}\ln^{2}{\Bigl(h^{-\beta}\Bigr)}}{16\pi^{2}y_{3}^{3}}+\frac{h^{\beta+1}}{8\pi y_{3}^{2}}\ln{\Bigl(8\pi y_{3}h^{-\beta}\Bigr)}+\frac{\varepsilon^{-2}h^{2\beta-1}}{4\pi^{2}y_{3}}
                                                                                                                 \le \frac{h^{\beta+1}}{4\pi y_{3}^{2}}\ln{\Bigl(2\pi y_{3}h^{-\beta}\Bigr)}+\frac{\varepsilon^{-2}h^{2\beta-1}}{2\pi^{2}y_{3}},
\end{split}
\end{equation*}
for $h$ sufficiently small.

Now we look at the real part of \cref{realimaginarypartzkbeta>1}
\begin{multline*}
\abs{\Re{z_{k}}+\frac{(2k+1)\pi h}{2y_{3}}}=\frac{h}{2y_{3}}\abs{\Arg{\biggr(1+\frac{i(2k+1)}{8y_{3}}h^{\beta}\biggr)}-\Im{R_{k}}}\\
\le\frac{h}{2y_{3}}\abs{\Arg{\biggr(1+\frac{i(2k+1)}{8y_{3}}h^{\beta}\biggr)}}+\frac{h}{2y_{3}}\abs{\frac{\ln{\Bigl(\ln{(w)}+(2k+1)i\pi\Bigr)}}{\ln{(w)}+(2k+1)i\pi}-R_{k}}
+\frac{h}{2y_{3}}\abs{\frac{\ln{\Bigl(\ln{(w)}+(2k+1)i\pi\Bigr)}}{\ln{(w)}+(2k+1)i\pi}}.
\end{multline*}
Using that $\abs{\Arg{(1+is)}}\le 2\abs{s}$ for small $s$, the second inequality in \cref{inequalityk}, \cref{bounditerativelogarithmbeta>1} and \cref{boundrestobeta>1} we get the bound
\begin{equation}
\label{boundRezkbeta>1}
\abs{\Re{z_{k}}+\frac{(2k+1)\pi h}{2y_{3}}}\le\frac{\varepsilon^{-1}h^{\beta}}{2\pi y_{3}}+\frac{h^{\beta+1}}{8y_{3}^{2}}+\frac{h^{2\beta+1}\ln^{2}{\Bigl(h^{-\beta}\Bigr)}}{16\pi^{2}y_{3}^{3}}+\frac{h^{\beta+1}\ln{\Bigl(h^{-\beta}\Bigr)}}{8\pi y_{3}^{2}}\le\frac{\varepsilon^{-1}h^{\beta}}{\pi y_{3}},
\end{equation}
the last inequality being valid for $h$ small enough.

This new inequality can be used to bound the following quantity 
\begin{equation*}
B=\abs{\frac{2y_{3}}{h(2k+1)^{2}\pi^{2}}\ln{\Bigl(8\pi y_{3}h^{-\beta}\Bigr)}\biggl(\Re{z_{k}}\biggr)^{2}-\frac{h}{2y_{3}}\ln{\Bigl(8\pi y_{3}h^{-\beta}\Bigr)}}.
\end{equation*}
By factorizing, \cref{boundRezkbeta>1} and by the second inequalities in \cref{inequalityk} and \cref{inequalityzk} respectively
\begin{equation*}
\begin{split}
B  \le \frac{2y_{3}}{h(2k+1)^{2}\pi^{2}}\ln{\Bigl(8\pi y_{3}h^{-\beta}\Bigr)}\Biggl(\abs{z_{k}}+\frac{\abs{2k+1}\pi h}{2y_{3}}\Biggr)\abs{\Re{z_{k}}+\frac{(2k+1)\pi h}{2y_{3}}}
   \le \frac{6\varepsilon^{-2}}{(2k+1)^{2}\pi^{3}}h^{\beta-1}\ln{\Bigl(8\pi y_{3}h^{-\beta}\Bigr)}.
\end{split}
\end{equation*}
Now we observe that $(2k+1)^{2}\ge k^{2}$ for all $ k\in\Z$. Using this fact, together with the reciprocal of the first inequality of \cref{inequalityk}, we have
\begin{equation*}
B\le\frac{24\varepsilon^{-4}}{\pi y_{3}^{2}}h^{\beta+1}\ln{\Bigl(8\pi y_{3}h^{-\beta}\Bigr)}.
\end{equation*}
Finally we can find the bound
\begin{multline*}
\abs{\Im{z_{k}+\frac{2y_{3}\ln{\Bigl(8\pi y_{3}h^{-\beta}\Bigr)}}{h(2k+1)^{2}\pi^{2}}\biggl(\Re{z_{k}}\biggr)^{2}}}\le\abs{\Im{z_{k}}+\frac{h}{2y_{3}}\ln{\Bigl(8\pi y_{3}h^{-\beta}\Bigr)}}\\
+\abs{\frac{2y_{3}\ln{\Bigl(8\pi y_{3}h^{-\beta}\Bigr)}}{h(2k+1)^{2}\pi^{2}}\biggl(\Re{z_{k}}\biggr)^{2}-\frac{h}{2y_{3}}\ln{\Bigl(8\pi y_{3}h^{-\beta}\Bigr)}}
\le\frac{1+96\varepsilon^{-4}}{4\pi y_{3}^{2}}h^{\beta+1}\ln{\Bigl(8\pi y_{3}h^{-\beta}\Bigr)}+\frac{\varepsilon^{-2}h^{2\beta-1}}{2\pi^{2}y_{3}}.
\end{multline*}
To get \cref{upperboundbeta>1} and \cref{lowerboundbeta>1} we use respectively that for all $k\in\Z$, $k\ne 0$, $k^{2}\le(2k+1)^{2}\le16k^{2}$, together with \cref{inequalityk}.
\end{proof}
Now we consider $\beta>1$ for \cref{semiclassicalequation}$^{-}$.
\begin{proposition}
\label{propositionbeta>1negative}
Let $\beta>1$ and $\varepsilon\in (0, 1)$ be given. Then there is $h_{0}>0$ such that, when $h\in(0, h_{0}]$, all solutions to \cref{semiclassicalequation}$^{-}$ satisfying
\begin{equation}
\label{inequalityzknegative2}
\varepsilon\le\abs{z}\le {1}/{\varepsilon}
\end{equation}
obey
\begin{equation*}
\abs{\Im{z}+\frac{y_{3}\ln{\Bigl(8\pi y_{3}h^{-\beta}\Bigr)}}{2k^{2}\pi^{2}h}\biggl(\Re{z}\biggr)^{2}}\le\frac{1+24\varepsilon^{-4}}{4\pi y_{3}^{2}}h^{\beta+1}\ln{\Bigl(8\pi y_{3}h^{-\beta}\Bigr)}+\frac{\varepsilon^{-2}h^{2\beta-1}}{2\pi^{2}y_{3}},
\end{equation*}
which, using \cref{inequalityk}, implies
\begin{gather*}
\Im{z}+\frac{\varepsilon^{2}}{8y_{3}}h\ln{\Bigl(8\pi y_{3}h^{-\beta}\Bigr)}\Bigl(\Re{z}\Bigr)^{2}\le\frac{1+24\varepsilon^{-4}}{4\pi y_{3}^{2}}h^{\beta+1}\ln{\Bigl(8\pi y_{3}h^{-\beta}\Bigr)}+\frac{\varepsilon^{-2}h^{2\beta-1}}{2\pi^{2}y_{3}}\\
\Im{z}+\frac{2}{\varepsilon^{2}y_{3}}h\ln{\Bigl(8\pi y_{3}h^{-\beta}\Bigr)}\Bigl(\Re{z}\Bigr)^{2}\ge-\frac{1+24\varepsilon^{-4}}{4\pi y_{3}^{2}}h^{\beta+1}\ln{\Bigl(8\pi y_{3}h^{-\beta}\Bigr)}-\frac{\varepsilon^{-2}h^{2\beta-1}}{2\pi^{2}y_{3}}.
\end{gather*}
\end{proposition}
\begin{proof}
Since $\beta>1$, by the second inequality in \cref{inequalityk}, we have
\begin{equation*}
\abs{2k+1}\pi\le 2\abs{k}\pi+\pi\le 4y_{3}\varepsilon^{-1}h^{-1}+\pi<8\pi y_{3}h^{-\beta}=\abs{\ln{(w)}},
\end{equation*}
for $h$ small enough. Recalling \cref{zkexpressionnegative}, we write
\begin{equation*}
z_{k}=\frac{ih}{2y_{3}}\Biggl(2i\pi k-\ln{\Bigl(8\pi y_{3}h^{-\beta}\Bigr)}-\ln{\biggl(1-\frac{i(2k+1)}{8y_{3}}h^{\beta}\biggr)}+R_{k}\Biggr)
\end{equation*}
and so we can ensure that the third term is small for $h$ small. Separating both the real and imaginary part gives
\begin{equation}
\label{realimaginarypartzkbeta>1negative}
\begin{cases}
\Re{z_{k}}=\frac{h}{2y_{3}}\Biggr(-2k\pi+\Arg{\biggr(1-\frac{i(2k+1)}{8y_{3}}h^{\beta}\biggr)}-\Im{R_{k}}\Biggr)\\
\Im{z_{k}}=\frac{h}{2y_{3}}\Biggr(-\ln{\Bigl(8\pi y_{3}h^{-\beta}\Bigr)}-\ln{\abs{1-\frac{i(2k+1)}{8y_{3}}h^{\beta}}}+\Re{R_{k}}\Biggr)
\end{cases}.
\end{equation}
We observe that
\begin{equation*}
\ln{\abs{1-\frac{i(2k+1)}{8y_{3}}h^{\beta}}}=\frac{1}{2}\ln\biggl(1+\frac{(2k+1)^{2}}{64y_{3}^{2}}h^{2\beta}\biggr)=\frac{1}{2}\ln{(1+t)},
\end{equation*}
with $t=\frac{(2k+1)^{2}}{64y_{3}^{2}}h^{2\beta}$. 

\cref{bounditerativelogarithmbeta>1} and \cref{boundrestobeta>1} are true and can be proved as done in \cref{propositionbeta>1}. So, by \cref{realimaginarypartzkbeta>1negative}
\begin{multline*}
\abs{\Im{z_{k}}+\frac{h}{2y_{3}}\ln{\Bigl(8\pi y_{3}h^{-\beta}\Bigr)}}\le \frac{h}{2y_{3}}\abs{\Re{R_{k}}-\Re{\Biggl(\frac{\ln{\Bigl(\ln{(w)}+(2k+1)i\pi\Bigr)}}{\ln{(w)}+(2k+1)i\pi}\Biggr)}}\\
+\frac{h}{2y_{3}}\abs{\Re{\Biggl(\frac{\ln{\Bigl(\ln{(w)}+(2k+1)i\pi\Bigr)}}{\ln{(w)}+(2k+1)i\pi}\Biggr)}-\ln{\abs{1-\frac{i(2k+1)}{8y_{3}}h^{\beta}}}}\le\frac{h^{2\beta+1}\ln^{2}{\Bigl(h^{-\beta}\Bigr)}}{16\pi^{2}y_{3}^{3}}+\abs{A},
\end{multline*}
where
\begin{equation*}
{A=\frac{-\biggl(\ln{\Bigl(8\pi y_{3}h^{-\beta}\Bigr)}+\frac{1}{2}\ln{(1+t)}\biggr)\frac{h^{\beta+1}}{16\pi y_{3}^{2}}+\frac{(2k+1)h^{2\beta+1}}{128\pi y_{3}^{3}}\Arg{\biggl(-8\pi y_{3}h^{-\beta}+(2k+1)i\pi\biggr)}}{1+t}-\frac{h}{4y_{3}}\ln{(1+t)}}.
\end{equation*}
By triangle inequality and $t>0$ it follows that
\begin{multline*}
\abs{A}\le\frac{h^{\beta+1}}{16\pi y_{3}^{2}}\ln{\Bigl(8\pi y_{3}h^{-\beta}\Bigr)}+\frac{\abs{2k+1}h^{2\beta+1}}{128\pi y_{3}^{3}}\abs{\Arg{\biggl(-8\pi y_{3}h^{-\beta}+(2k+1)i\pi\biggr)}}\\
+\frac{h}{4y_{3}}\ln{(1+t)}+\frac{h^{\beta+1}}{32\pi y_{3}^{2}}\ln{(1+t)}.
\end{multline*}
This expression can be also bounded using $\abs{\Arg{\biggl(8\pi y_{3}h^{-\beta}+(2k+1)i\pi\biggr)}}\le\pi$, $\ln{(1+t)}\le t$, the expression for $t$, $(2k+1)^{2}\le2(4k^{2}+1)$ and the second inequality in \cref{inequalityk} as
\begin{equation*}
\abs{A}\le\frac{h^{\beta+1}}{16\pi y_{3}^{2}}\ln{\Bigl(8\pi y_{3}h^{-\beta}\Bigr)}+\frac{\varepsilon^{-2}h^{3\beta-1}}{64\pi^{3}y_{3}^{2}}+\frac{h^{3\beta+1}}{1024\pi y_{3}^{4}}+\frac{\varepsilon^{-1}h^{2\beta}}{32\pi y_{3}^{2}}+\frac{h^{2\beta+1}}{64y_{3}^{3}}+\frac{\varepsilon^{-2}h^{2\beta-1}}{8\pi^{2}y_{3}}.
\end{equation*}
For $\beta>1$, the dominant term for small $h$ are the first and last term. More specifically the first is dominant for $\beta\ge 2$ and the other for $1<\beta<2$. What said implies that for $h$ small enough
\begin{equation*}
\abs{A}\le\frac{h^{\beta+1}}{8\pi y_{3}^{2}}\ln{\Bigl(8\pi y_{3}h^{-\beta}\Bigr)}+\frac{\varepsilon^{-2}h^{2\beta-1}}{4\pi^{2}y_{3}},
\end{equation*}
which implies
\begin{equation*}
\begin{split}
\abs{\Im{z_{k}}+\frac{h}{2y_{3}}\ln{\Bigl(8\pi y_{3}h^{-\beta}\Bigr)}} \le \frac{h^{2\beta+1}\ln^{2}{\Bigl(h^{-\beta}\Bigr)}}{16\pi^{2}y_{3}^{3}}+\frac{h^{\beta+1}}{8\pi y_{3}^{2}}\ln{\Bigl(8\pi y_{3}h^{-\beta}\Bigr)}+\frac{\varepsilon^{-2}h^{2\beta-1}}{4\pi^{2}y_{3}}
                                                                                                                \le \frac{h^{\beta+1}}{4\pi y_{3}^{2}}\ln{\Bigl(8\pi y_{3}h^{-\beta}\Bigr)}+\frac{\varepsilon^{-2}h^{2\beta-1}}{2\pi^{2}y_{3}},
\end{split}
\end{equation*}
for $h$ sufficiently small.

Now we look at the real part of \cref{realimaginarypartzkbeta>1negative}
\begin{multline*}
\abs{\Re{z_{k}}+\frac{k\pi h}{y_{3}}}=\frac{h}{2y_{3}}\abs{\Arg{\biggr(1-\frac{i(2k+1)}{8y_{3}}h^{\beta}\biggr)}-\Im{R_{k}}}\\
\le\frac{h}{2y_{3}}\abs{\Arg{\biggr(1-\frac{i(2k+1)}{8y_{3}}h^{\beta}\biggr)}}+\frac{h}{2y_{3}}\abs{\frac{\ln{\Bigl(\ln{(w)}+(2k+1)i\pi\Bigr)}}{\ln{(w)}+(2k+1)i\pi}-R_{k}}
+\frac{h}{2y_{3}}\abs{\frac{\ln{\Bigl(\ln{(w)}+(2k+1)i\pi\Bigr)}}{\ln{(w)}+(2k+1)i\pi}}.
\end{multline*}
Using that $\abs{\Arg{(1+is)}}\le 2\abs{s}$ for small $s$, the second inequality in \cref{inequalityk}, \cref{bounditerativelogarithmbeta>1} and \cref{boundrestobeta>1} we get the bound
\begin{equation}
\label{boundRezkbeta>1negative}
\abs{\Re{z_{k}}+\frac{k\pi h}{y_{3}}}\le\frac{\varepsilon^{-1}h^{\beta}}{2\pi y_{3}}+\frac{h^{\beta+1}}{8y_{3}^{2}}+\frac{h^{2\beta+1}\ln^{2}{\Bigl(h^{-\beta}\Bigr)}}{16\pi^{2}y_{3}^{3}}+\frac{h^{\beta+1}\ln{\Bigl(h^{-\beta}\Bigr)}}{8\pi y_{3}^{2}}\le\frac{\varepsilon^{-1}h^{\beta}}{\pi y_{3}},
\end{equation}
last inequality being valid for $h$ small enough.

This new inequality can be used to bound the quantity 
\begin{equation*}B=\abs{\frac{y_{3}}{2k^{2}\pi^{2}h}\ln{\Bigl(8\pi y_{3}h^{-\beta}\Bigr)}\biggl(\Re{z_{k}}\biggr)^{2}-\frac{h}{2y_{3}}\ln{\Bigl(8\pi y_{3}h^{-\beta}\Bigr)}}.
\end{equation*}
By factorizing, \cref{boundRezkbeta>1negative} and by the inequalities \cref{inequalityk} and \cref{inequalityzknegative2} 
\begin{equation*}
\begin{split}
B & \le \frac{y_{3}}{2k^{2}\pi^{2}h}\ln{\Bigl(8\pi y_{3}h^{-\beta}\Bigr)}\Biggl(\abs{z_{k}}+\frac{\abs{k}\pi h}{y_{3}}\Biggr)\abs{\Re{z_{k}}+\frac{k\pi h}{y_{3}}}\\
  & \le \frac{6\varepsilon^{-4}}{\pi y_{3}^{2}}h^{\beta+1}\ln{\Bigl(8\pi y_{3}h^{-\beta}\Bigr)}.
\end{split}
\end{equation*}
Finally we can find the desired bound
\begin{align*}
\abs{\Im{z_{k}+\frac{y_{3}\ln{\Bigl(8\pi y_{3}h^{-\beta}\Bigr)}}{2k^{2}\pi^{2}h}\biggl(\Re{z_{k}}\biggr)^{2}}}&\le\abs{\Im{z_{k}}+\frac{h}{2y_{3}}\ln{\Bigl(8\pi y_{3}h^{-\beta}\Bigr)}}
+\abs{\frac{2y_{3}\ln{\Bigl(8\pi y_{3}h^{-\beta}\Bigr)}}{h(2k+1)^{2}\pi^{2}}\biggl(\Re{z_{k}}\biggr)^{2}-\frac{h}{2y_{3}}\ln{\Bigl(8\pi y_{3}h^{-\beta}\Bigr)}}\\
&\le\frac{1+24\varepsilon^{-4}}{4\pi y_{3}^{2}}h^{\beta+1}\ln{\Bigl(8\pi y_{3}h^{-\beta}\Bigr)}+\frac{\varepsilon^{-2}h^{2\beta-1}}{2\pi^{2}y_{3}}.
\end{align*}
\end{proof}
\subsection{Neumann Boundary Condition}
By the definition of $-h^{2}\Delta_{\pm h^{-\beta}, y}^{\textup{N}}$ it follows that the resonances are the solutions of
\begin{equation}
\label{semiclassicalequationN}
\pm h^{-\beta}-\frac{iz}{4\pi h}-\frac{e^{2iy_{3}\frac{z}{h}}}{8\pi y_{3}}=0,
\end{equation}
which can be rewritten as $w_{\pm}=x_{\pm}e^{x_{\pm}}$, where $x_{\pm}=\pm 8\pi y_{3}h^{-\beta}-2iy_{3}\frac{z}{h}$ and $w_{\pm}=e^{\pm 8\pi y_{3}h^{-\beta}}$. So the solutions are still expressed in terms of the Lambert W function, with the only difference of it having a positive argument this time. A corresponding asymptotic expansion is still valid (one has only to substitute $2k+1$ with $2k$, reason being the change in sign of $w$)
\begin{gather}
W_{k}(y)= \ln{(w)}+2ik\pi -\ln{\Bigl(\ln{(w)}+2ik\pi\Bigr)}+R_{k}\nonumber\\
\label{seriesLambertN}
R_{k}=\sum_{j=0}^{+\infty}\sum_{m=1}^{+\infty}c_{j, m}\frac{\ln^{m}{\Bigl(\ln{(w)}+2ik\pi \Bigr)}}{\Bigl(\ln{(w)}+2ik\pi\Bigr)^{j+m}}.
\end{gather}

The expression for the resonances $z_{k}^{\pm}$ in term of the Lambert $W$ function is
\begin{equation}
\label{zkexpressionN}
z_{k}^{\pm}= \frac{ih}{2y_3}\biggl(2ik\pi-\ln{\Bigl(\ln{(y_{\pm})}+2ik\pi\Bigr)}+R_{k}\biggr).
\end{equation}

Properties of semi-classic asymptotic are here reported without explicit proofs, amounting to slight modifications of the ones of the Dirichlet case.
\begin{lemma}
\label{lemmaboundresto2}
The series \cref{seriesLambertN} is absolutely convergent for $w$ large (or small) enough and $k\in\Z$. More precisely we have the tail estimate
\begin{equation*}
\begin{split}
\abs{R_{k}-\frac{\ln{\Bigl(\ln{(w)}+2ik\pi\Bigr)}}{\ln{(w)}+2ik\pi}} \le \sum_{j\ge 0, m\ge 1, (j, m)\ne (0, 1)}\abs{c_{j, m}\frac{\ln^{m}{\Bigl(\ln{(w)}+2ik\pi\Bigr)}}{\Bigl(\ln{(w)}+2ik\pi\Bigr)^{j+m}}}
                                                                                                            \le 2\abs{\frac{\ln{\Bigl(\ln{(w)}+2ik\pi\Bigr)}}{\ln{(w)}+2ik\pi}}^{2}
\end{split}
\end{equation*}
and
\begin{equation*}
\abs{\frac{\ln{\Bigl(\ln{(w)}+2ik\pi\Bigr)}}{\ln{(w)}+2ik\pi}}\le\frac{1}{2}.
\end{equation*}
\end{lemma}
\begin{lemma}
\label{lemmainequalitykN}
Let $\varepsilon\in (0, 1)$ ve given. Then, for $k$ such that $z_{k}$ is given by \cref{zkexpressionN}  and 
$\varepsilon\le\abs{z_{k}}\le1/\varepsilon$
we have
\begin{equation}
\label{inequalitykN}
\frac{\varepsilon}{2}\le\frac{\abs{k}\pi h}{y_{3}}\le\frac{2}{\varepsilon}.
\end{equation}
\end{lemma}
\begin{proposition}
Let $\beta\in(0, 1)$ and $\varepsilon\in(0, 1)$ be given. Then there is $h_{0}>0$ such that, when $h\in(0, h_{0}]$, all solutions of \cref{semiclassicalequationN} satisfying
$\varepsilon\le\abs{z}\leq {1}/{\varepsilon}$
obey
\begin{equation*}
0\le-\Im{z}-\frac{h}{2y_{3}}\ln{\Bigl(2y_{3}h^{-1}\abs{\Re{z}}\Bigr)}\le\frac{72\pi^{2}}{y_{3}}\varepsilon^{-2}h^{3-2\beta}.
\end{equation*}
\end{proposition}
\begin{proposition}
\label{propositionbeta>1N}
Let $\beta>1$ and $\varepsilon\in (0, 1)$ be given. Then there is $h_{0}>0$ such that, when $h\in(0, h_{0}]$, all solutions to \cref{semiclassicalequationN}$^{+}$ satisfying
$\varepsilon\le\abs{z}\leq {1}/{\varepsilon}$
obey
\begin{equation*}
\abs{\Im{z_{k}+\frac{2y_{3}}{hk^{2}\pi^{2}}\ln{\Bigl(8\pi y_{3}h^{-\beta}\Bigr)}\biggl(\Re{z_{k}}\biggr)^{2}}}\le\frac{1+96\varepsilon^{-4}}{4\pi y_{3}^{2}}h^{\beta+1}\ln{\Bigl(8\pi y_{3}h^{-\beta}\Bigr)}+\frac{\varepsilon^{-2}h^{2\beta-1}}{4\pi^{2}y_{3}},
\end{equation*}
which, using \cref{inequalitykN}, implies
\begin{gather*}
\Im{z}+\frac{\varepsilon^{2}}{2y_{3}}h\ln{\Bigl(8\pi y_{3}h^{-\beta}\Bigr)}\Bigl(\Re{z}\Bigr)^{2}\le\frac{1+96\varepsilon^{-4}}{4\pi y_{3}^{2}}h^{\beta+1}\ln{\Bigl(8\pi y_{3}h^{-\beta}\Bigr)}+\frac{\varepsilon^{-2}h^{2\beta-1}}{4\pi^{2}y_{3}}\\
\Im{z}+\frac{8}{\varepsilon^{2}y_{3}}h\ln{\Bigl(8\pi y_{3}h^{-\beta}\Bigr)}\Bigl(\Re{z}\Bigr)^{2}\ge-\frac{1+96\varepsilon^{-4}}{4\pi y_{3}^{2}}h^{\beta+1}\ln{\Bigl(8\pi y_{3}h^{-\beta}\Bigr)}-\frac{\varepsilon^{-2}h^{2\beta-1}}{4\pi^{2}y_{3}}.
\end{gather*}
\end{proposition}
\begin{proposition}
\label{propositionbeta>1negative}
Let $\beta>1$ and $\varepsilon\in (0, 1)$ be given. Then there is $h_{0}>0$ such that, when $h\in(0, h_{0}]$, all solutions to \cref{semiclassicalequationN}$^{-}$ satisfying
\begin{equation*}
\label{inequalityzknegative2N}
\varepsilon\le\abs{z}\leq {1}/{\varepsilon}
\end{equation*}
obey
\begin{equation*}
\abs{\Im{z_{k}+\frac{2y_{3}\ln{\Bigl(8\pi y_{3}h^{-\beta}\Bigr)}}{(2k-1)^{2}\pi^{2}h}\biggl(\Re{z_{k}}\biggr)^{2}}}\le\frac{1+96\varepsilon^{-4}}{4\pi y_{3}^{2}}h^{\beta+1}\ln{\Bigl(8\pi y_{3}h^{-\beta}\Bigr)}+\frac{\varepsilon^{-2}h^{2\beta-1}}{4\pi^{2}y_{3}},
\end{equation*}
which, using \cref{inequalitykN}, implies
\begin{gather*}
\Im{z}+\frac{\varepsilon^{2}}{32y_{3}}h\ln{\Bigl(8\pi y_{3}h^{-\beta}\Bigr)}\Bigl(\Re{z}\Bigr)^{2}\le\frac{1+24\varepsilon^{-4}}{4\pi y_{3}^{2}}h^{\beta+1}\ln{\Bigl(8\pi y_{3}h^{-\beta}\Bigr)}+\frac{\varepsilon^{-2}h^{2\beta-1}}{4\pi^{2}y_{3}}\\
\Im{z}+\frac{8}{\varepsilon^{2}y_{3}}h\ln{\Bigl(8\pi y_{3}h^{-\beta}\Bigr)}\Bigl(\Re{z}\Bigr)^{2}\ge-\frac{1+24\varepsilon^{-4}}{4\pi y_{3}^{2}}h^{\beta+1}\ln{\Bigl(8\pi y_{3}h^{-\beta}\Bigr)}-\frac{\varepsilon^{-2}h^{2\beta-1}}{4\pi^{2}y_{3}}.
\end{gather*}
\end{proposition}

\section{Resonance Expansion for the wave and Schr\"odinger Equation with Point Interaction}
In the present Section, we prove the resonance expansion for the solution of the wave equation and for the propagator of the Schr\"odinger equation on the half-space in the presence of a point interaction.\\ As regards the wave equation, a similar analysis in the presence of a smooth $L_{\textup{comp}}^{\infty}(\R^{n})$ potential on the whole space is given in \cite{DZ:MathematicalTheoryScatteringResonances} (pages 39-45 and 110-111). We recall that $L^{\infty}_{\textup{comp}}(\R^n)$ is the space of measurable, essentially bounded functions with compact support. In the following Theorem we will treat the abstract wave equation with generator given by a point interaction in $\R^3_+$ and initial data in $L^{2}_{\textup{comp}}(\R^3_+)$ and $H^{1}_{\textup{comp}}(\R^3_+)$, the space of square integrable function in $\R^3_+$ with compact support and the corresponding Sobolev space of order one. The restriction to data with compact support is needed, as in the standard case, to tame the exponentially diverging behavior of resonance at spatial infinity.
\begin{theorem}
\label{resonanceexpansionthm}
Let $\alpha\in\R$, $y=(y_{1}, y_{2}, y_{3})\in\R_{+}^{3}$ and 
suppose $w(t)$ is the solution of
\begin{equation}
\label{waveequationproblem}
\begin{cases}
\biggl(\frac{\partial^{2}}{\partial t^{2}}-\Delta_{\alpha, y}^{\textup{D}}\biggr)w(t)=0\\
w(0)=w_{0}\in H_{\textup{comp}}^{1}(\R_{+}^{3})\\
\frac{\partial w}{\partial t}(0)=w_{1}\in L_{\textup{comp}}^{2}(\R_{+}^{3})
\end{cases}.
\end{equation}
Let $\Set{z_{j}}$ be respectively the set of resonances including the possible eigenvalue of $\Delta_{\alpha, y}^{\textup{D}}$, and $R(z):=\Bigl(-\Delta_{\alpha, y}^{\textup{D}}-z^{2}\Bigr)^{-1}$
Then, for any $A>0$,
\begin{itemize}
\item[i)] if $\alpha\ne -\frac{1}{8\pi y_{3}}$
\begin{equation*}
w(t)=\sum_{\set{\Im{z_{j}}>-A}}e^{-iz_{j}t}f_{j}+\quad E_{A}(t),
\end{equation*}
where the sum is finite and
\begin{equation}
\label{definitionfj}
f_{j}=-\Res_{z=z_{j}}{\biggl(iR(z)w_{1}+zR(z)w_{0}\biggr)};
\end{equation}
\item[ii)] if $\alpha =-\frac{1}{8\pi y_{3}}$
\begin{equation*}
w(t)=\sum_{\set{\Im{z_{j}}>-A}}e^{-iz_{j}t}f_{j}+f_{0, 0}+tf_{0, 1}+\quad E_{A}(t),
\end{equation*}
where the sum is finite, $f_{j}$ are still defined as in \cref{definitionfj} and
\begin{equation*}
f_{0, 0}+tf_{0, 1}=-\Res_{z=0}{\biggl(iR(z)w_{1}+zR(z)w_{0}\biggr)}.
\end{equation*}
\end{itemize}
Moreover, for any compact, connected set $K$ such that $\supp{(w_{j})}\subset K$, $j=0, 1$, there exist constants $C_{K, A}$ and $T_{K, A}$ such that
\begin{equation*}
\norma{E_{A}(t)}_{H^{2}(K)}\le C_{K, A}e^{-tA}(\norma{w_{0}}_{H^{1}(K)}+\norma{w_{1}}_{L^{2}(K)}), \quad t\ge T_{K, A}.
\end{equation*}
\end{theorem}
The following technical lemma allows to bound on the meromorphically continued resolvent on $\C_{-}$ when smoothly truncated by means of a cutoff function $\rho$.
\begin{lemma}
Let $\alpha\in\R$ and $y=(y_{1}, y_{2}, y_{3})\in\R_{+}^{3}$ and $K\subset\R_{+}^{3}$ compact. For any $\rho\in C_{0}^{\infty}(\R_{+}^{3})$ there exist constants $A, C$ and $T$ depending on $\rho$ such that
\begin{equation}
\label{boundtruncatedresolvent}
\norma{\rho(-\Delta_{\alpha, y}^{\textup{D}}-z^{2})^{-1}\rho}_{L^{2}(K)\to H^{j}(K)}\le C(1+\abs{z})^{j-1}e^{T(\Im{z})_{-}}, \ \ \ j=0,1,2
\end{equation}
for
\begin{equation*}
\Im{z}\ge -A-\delta\ln{(1+\abs{z})}, \quad\abs{z}>C_{0}, \quad \delta<\frac{1}{2y_{3}}
\end{equation*}
where $z$ is not a resonance and it is denoted $(x)_{-}=\max\Set{-x, 0}$.\\
\end{lemma}
\begin{proof}
We start by proving that for any $A>0$ there are only finitely many resonances in the region 
\begin{equation}
\label{regionfiniteresonances}
\Set{z\in\C|\Im{z}\ge -A-\delta\ln{(1+\abs{z})}}\ .
\end{equation}
 As shown previously, all resonances are on the curve \cref{bintermsofahalfspace2}. This means that the curve
\begin{equation}\label{auxcurve}\Im{z}=\frac{1}{2y_{3}}\ln{\Biggl(\frac{1}{2y_{3}\abs{\Re{z}}}\Biggr)} \ \ \ \ \ \text{equivalent to}\ \ \ \ 
\abs{\Re{z}}=\frac{1}{2y_{3}}e^{-2y_{3}\Im{z}}
\end{equation}
is above all resonances. The equation of the curve which delimits the region \cref{regionfiniteresonances} can be written in explicit form as follows
\begin{equation}
\label{curveregionfiniteresonanceexplicit}
\abs{\Re{z}}=\sqrt{\biggl(e^{-\frac{\Im{z}+A}{\delta}}-1\biggr)^{2}-\Bigl(\Im{z}\Bigr)^{2}}.
\end{equation}
Both curves are symmetric with respect to the imaginary axis, so we limit ourself to $\Re{z}>0$. If it happens that for $\Im{z}$ large enough, curve \cref{curveregionfiniteresonanceexplicit} is above curve \cref{auxcurve}, then for $\Im{z}$ large enough all $z$ which generate a resonance and obey this constraint are below \cref{curveregionfiniteresonanceexplicit} and hence only a finite quantity of them is in the region described by \cref{regionfiniteresonances}. This scenario happens for $\frac{1}{\delta}>2y_{3}$ regardless of the value of $A$.

In order to prove \cref{boundtruncatedresolvent} we use \cref{dirichletoperator} and write 
\begin{multline}
\label{triangleinequalityresolvent}
\norma{\rho(-\Delta_{\alpha, y}^{\textup{D}}-z^{2})^{-1}\rho}_{L^{2}(K)\to H^{j}(K)}\le \norma{\rho(-\Delta^{\textup{D}}-z^{2})^{-1}\rho}_{L^{2}(K)\to H^{j}(K)}\\
+\abs{\Gamma_{\alpha, y}^{\textup{D}}(z)}^{-1}\norma{\rho G_{z, y}^{\textup{D}}(-\Delta^{\textup{D}}-z^{2})^{-1}\rho}_{L^{2}(K)\to H^{j}(K)}.
\end{multline}
We recall the following bound for the free truncated resolvent (\cite{DZ:MathematicalTheoryScatteringResonances} Theorem 3.1)
\begin{equation*}
\norma{\rho(-\Delta-z^{2})^{-1}\rho}_{L^{2}(K)\to H^{j}(K)}\le C(1+\abs{z})^{j-1}e^{T(\Im{z})_{-}},\quad j=0, 1, 2,
\end{equation*}
where $L>\diam{(\supp{(\rho)})}$ and that is also valid if we replace $-\Delta$ with $-\Delta^{\textup{D}}$. This directly bounds the first term in \cref{triangleinequalityresolvent}. It does the same for the second one if we observe that $G_{z, y}^{\textup{D}}$ is in $L_{\textup{loc}}^{2}(\R_{+}^{3})$ and that outside of resonances $\abs{\Gamma_{\alpha, y}^{\textup{D}}(z)}^{-1}$ is bounded. Hence \cref{boundtruncatedresolvent} follows.
\end{proof}
\begin{proof}[Proof of \cref{resonanceexpansionthm}]
We start by considering the case $w_{0}=0$ and $\supp{(w_{1})}\subset K$.

By the functional calculus, the propagator of \cref{waveequationproblem} can be written as 
\begin{equation}
\label{definitionU}
U(t)=\int_{0}^{+\infty}\frac{\sin{tz}}{z}\diff E_{z}+\frac{\sin{t\mu}}{\mu}\biggl(\cdot, G_{\mu, y}^{\textup{D}}\biggr)G_{\mu, y}^{\textup{D}},
\end{equation}
where $-\mu^{2}$ is the possible eigenvalue of $-\Delta_{\alpha, y}^{\textup{D}}$ (the term is not present when there is no eigenvalue).If the eigenvalue is zero (case $ii)$ in the statement), the last term is replaced by
\begin{equation*}
t\,\biggl(\cdot, G_{0, y}^{\textup{D}}\biggr)G_{0, y}^{\textup{D}}.
\end{equation*}

{Since for $z$ near $\mu\ne 0$, $R(z)=(-\mu^{2}-z^{2})^{-1}\biggl(\cdot, G_{\mu, y}^{\textup{D}}\biggr)G_{\mu, y}^{\textup{D}}+Q(z)$, where $Q(z)$ is holomorphic near $-\mu^{2}$ and $\Res_{z=i\mu}{(-\mu^{2}-z^{2})}=-(2\mu)^{-1}$, we have
\begin{equation}
\label{eigenfunctionterm}
\frac{\sin{t\mu}}{\mu}\biggl(\cdot, G_{\mu, y}^{\textup{D}}\biggr)G_{\mu, y}^{\textup{D}}=-\Res_{z=i\mu}{(iR(z)e^{-izt})}-\Res_{z=-i\mu}{(iR(-z)e^{-izt})}=\Pi_{\mu}+\Pi_{-\mu}.
\end{equation}
If $\mu=0$, instead
\begin{equation}
\label{eigenfunctionterm0}
t\,\biggl(\cdot, G_{0, y}^{\textup{D}}\biggr)G_{0, y}^{\textup{D}}=-t\Res_{z=0}{(iR(z))}.
\end{equation}
In the following, the term for $\mu=0$ will not be reported whenever it is not needed to distinguish the two cases.

By Stone's formula for the spectral measure $\diff E_{z}$ in terms of $R(z)$, we have
\begin{equation*}
\diff E_{z}=\frac{1}{\pi i}(R(z)-R(-z))z\diff z.
\end{equation*}
Hence
\begin{equation*}
\begin{split}
w(t, x)-\frac{\sin{t\mu}}{\mu}\biggl(\cdot, G_{\mu, y}^{\textup{D}}\biggr)G_{\mu, y}^{\textup{D}} & = \frac{1}{\pi i}\int_{0}^{+\infty}\sin{tz}(R(z)-R(-z))w_{1}(x)\diff z\\
                 & = \frac{1}{2\pi}\lim_{\varepsilon\to 0^{+}}\int_{\R\setminus(-\varepsilon, \varepsilon)}e^{-itz}(R(z)-R(-z))w_{1}(x)\diff z\\
                 &
\begin{multlined}
=\frac{1}{2\pi}\int_{\Sigma_{\varepsilon_{0}}}e^{-itz}(R(z)-R(-z))w_{1}(x)\diff z\\
+\frac{1}{2\pi}\lim_{\varepsilon\to 0^{+}}\int_{\sigma_{\varepsilon}}e^{-itz}(R(z)-R(-z))w_{1}(x)\diff z,
\end{multlined}
\end{split}
\end{equation*}
where $\varepsilon_{0}$ is such that there are no non-zero poles of $R(z)$ in $D(0, \varepsilon_{0})$, $\Sigma_{\varepsilon_{0}}$ is the union of $\R\setminus(-\varepsilon_{0}, \varepsilon_{0})$ with the semicircle in the upper complex half-plane of radius $\varepsilon_{0}$ and centered at zero oriented counterclockwise. $\sigma_{\varepsilon}$ instead is the semicircle in the upper complex half-plane of radius $\varepsilon$, centered at zero and oriented clockwise.

To prove that the integrand decays fast enough for the integral on $\Sigma_{\varepsilon_{0}}$ to be convergent we observe that, by the definition of $R(z)$, it follows that
\begin{equation*}
\Bigl(R(z)-R(-z)\Bigr)\biggl(-\Delta_{\alpha, y}^{\textup{D}}\biggr)=(R(z)-R(-z))\biggl(-\Delta_{\alpha, y}^{\textup{D}}-z^{2}+z^{2}\biggr)=z^{2}\Bigl(R(z)-R(-z)\Bigr).
\end{equation*}
From this we can conclude that for $\rho\in C_{0}^{\infty}(\R_{+}^{3})$ equal to one on $\supp{(w_{1})}$,
\begin{equation*}
\begin{split}
\rho\Bigl(R(z)-R(-z)\Bigr)\biggl(1-\Delta_{\alpha, y}^{\textup{D}}\biggr)\rho w_{1} & = \rho\Bigl(R(z)-R(-z)\Bigr)w_{1}+\rho\Bigl(R(z)-R(-z)\Bigr)\biggl(-\Delta_{\alpha, y}^{\textup{D}}\biggr)w_{1}\\
                                                                                                                                                   & = (1+z^{2})\rho\Bigl(R(z)-R(-z)\Bigr)w_{1}
\end{split}
\end{equation*}
and so
\begin{equation}
\label{identityresolvent}
\rho\Bigl(R(z)-R(-z)\Bigr)\rho w_{1}=\rho\Bigl(R(z)-R(-z)\Bigr)(1+z^{2})^{-1}\biggl(1-\Delta_{\alpha, y}^{\textup{D}}\biggr)\rho w_{1}.
\end{equation}
Since $\rho\Bigl(R(z)-R(-z)\Bigr)\rho=O(|z|^{-2})$ as $|z|\to\infty$ as an operator $L^{2}(\R_{+}^{3})\to L^{2}(\R_{+}^{3})$, this shows that the integral converges in $L_{\textup{loc}}^{2}(\R_{+}^{3})$.

The resolvent is holomorphic in a neighborhhod of the origin and hence the integral over $\sigma_{\varepsilon}$ converges to $0$ as $\varepsilon\to 0^{+}$, unless zero is an eigenvalue for $-\Delta_{\alpha, y}^{\textup{D}}$. In that case, we study the expression $R(z)-R(-z)$ near zero. We recall that $0$ is an eigenvalue if and only if $\alpha=-\frac{1}{8\pi y_{3}}$. In this case
\begin{equation*}
R(z)-R(-z)=\frac{e^{iz\abs{x-y}}}{4\pi\abs{x-y}}\frac{e^{iz\abs{x'-y}}}{4\pi\abs{x'-y}}\Bigl(\Gamma_{\alpha, y}^{\textup{D}}(z)\Bigr)^{-1}-\frac{e^{-iz\abs{x-y}}}{4\pi\abs{x-y}}\frac{e^{-iz\abs{x'-y}}}{4\pi\abs{x'-y}}\Bigl(\Gamma_{\alpha, y}^{\textup{D}}(-z)\Bigr)^{-1}.
\end{equation*}
Some computations show that
\begin{equation*}
\Res_{z=0}{\Bigl(R(z)-R(-z)\Bigr)}=\frac{i}{\pi y_{3}\abs{x-y}\abs{x'-y}}\Biggl(\frac{y_{3}}{3}-\frac{\abs{x-y}+\abs{x'-y}}{2}\Biggr),
\end{equation*}
which means that
\begin{equation*}
\begin{split}
\frac{1}{2\pi}\lim_{\varepsilon\to 0^{+}}\int_{\sigma_{\varepsilon}}e^{-itz}(R(z)-R(-z))w_{1}(x)\diff z & = \frac{w_{1}(x)}{2\pi y_{3}\abs{x-y}\abs{x'-y}}\Biggl(\frac{y_{3}}{3}-\frac{\abs{x-y}+\abs{x'-y}}{2}\Biggr)\\
                                                                                                                                                      & = \Pi_{0}w_{1}(x).
\end{split}
\end{equation*}

Let $\rho\in C_{0}^{\infty}(\R_{+}^{3})$ satisfy $\rho=1$ on $K$. We choose $r$ large enough so that all the resonances $z^{2}$ with $\Im{z}>-A-\delta\ln{(1+\abs{\Re{z}})}$ are contained in the ball $\abs{z}\le r$. We deform the contour of integration in the integral over $\Sigma_{\varepsilon_{0}}$ using the following contours:
\begin{gather*}
\omega_{r}=\Set{s-i\,\Bigl(A+\varepsilon+\delta\ln{(1+\abs{s})}\Bigr)|s\in [-r, r]}\\
\gamma_{r}^{\pm}=\Set{\pm r-is|0\le s\le A+\varepsilon+\delta\ln{(1+r)}}, \quad \gamma_{r}=\gamma_{r}^{+}\cup\gamma_{r}^{-}\\
\gamma_{r}^{\infty}=(-\infty, r)\cup (r, +\infty).
\end{gather*}
$\varepsilon$ is chosen in such a way that $\omega_{r}$ does not cross resonances. We also define
\begin{equation*}
\Omega_{A}=\Set{z|\Im{z}\ge -A-\varepsilon-\delta\ln{(1+\abs{\Re{z}})}}\setminus\Set{0}
\end{equation*}
and define
\begin{equation*}
\Pi'_{A}(t)=\sum_{\Set{z_{j}\in\Omega_{A}}}\Res_{z=z_{j}}{\Bigl(\rho R(z)\rho e^{-izt}\Bigr)}.
\end{equation*}
By combining \cref{definitionU}, \cref{eigenfunctionterm} (or \cref{eigenfunctionterm0} when zero is an eigenvalue) and the residue theorem, we can write that
\begin{equation}
\label{decompositionU}
\rho U(t)\rho=\Pi'_{A}(t)+E_{\omega_{r}}(t)+E_{\gamma_{r}}(t)+E_{\gamma_{r}^{\infty}}(t)+\Pi_{\mu}+\Pi_{-\mu},
\end{equation}
where
\begin{equation*}
E_{\gamma}(t)=\frac{1}{2\pi}\int_{\gamma}e^{-itz}\rho\Bigl(R(z)-R(-z)\Bigr)\rho w_{1}\diff z.
\end{equation*}
We note that $\Pi_{-\mu}$ cancels out with the term in $\Pi'_{A}$ containing the residue at $z=-\mu$. So, by putting together this cancellations in the new term, $\Pi_{A}$, \cref{decompositionU} becomes
\begin{equation*}
\rho U(t)\rho=\Pi_{A}(t)+E_{\omega_{r}}(t)+E_{\gamma_{r}}(t)+E_{\gamma_{r}^{\infty}}(t).
\end{equation*}

Now we prove that the the last two terms in \cref{decompositionU} are negligible in the limit $r\to +\infty$. To do so, we consider $w_{1}\in H^{2}(K)$. This allows us to use \cref{identityresolvent}  to \cref{boundtruncatedresolvent} with $j=1$ and obtain the following bound
\begin{equation*}
\begin{split}
\norma{E_{\gamma_{r}^{\infty}}(t)w_{1}}_{H^{1}(K)} & \le C\norma{\int_{r}^{+\infty}e^{-itz}\rho\Bigl(R(z)-R(-z)\Bigr)\rho w_{1}\diff z}_{H^{1}(K)}\\
                                                                               & \le C\norma{\int_{r}^{+\infty}\rho\Bigl(R(z)-R(-z)\Bigr)(1+z^{2})^{-1}\Bigl(1-\Delta_{\alpha, y}^{\textup{D}}\Bigr)\rho w_{1}\diff z}_{H^{1}(K)}\\
                                                                               & \le C \int_{r}^{+\infty}(1+z^{2})^{-1}\norma{w_{1}}_{H^{2}(K)}\diff z=C\biggl(\frac{\pi}{2}-\arctan{r}\biggr)\norma{w_{1}}_{H^{2}(K)}\xrightarrow[r\to+\infty]{}0.
\end{split}
\end{equation*}
Instead, for the other term and for $t>T$ one has
\begin{equation*}
\begin{split}
\norma{E_{\gamma_{r}}(t)w_{1}}_{H^{1}(K)} & = C\norma{-i\int_{0}^{A+\varepsilon+\delta\ln{(1+r)}}e^{-it(r-is)}\rho\Bigl(R(r-is)-R(-r+is)\Bigr)\rho w_{1}\diff s}_{H^{1}(K)}\\
                                                                                            & \textup{\footnotesize $\le C\norma{\int_{0}^{A+\varepsilon+\delta\ln{(1+r)}}e^{-it(r-is)}\rho\Bigl(R(r-is)-R(r+is)\Bigr)(1+(r-is)^{2})^{-1}\Bigl(1-\Delta_{\alpha, y}^{\textup{D}}\Bigr)\rho w_{1}\diff s}_{H^{1}(K)}$}\\
                                                                                            & \le C\int_{0}^{A+\varepsilon+\delta\ln{(1+r)}}(1+(r-is)^{2})^{-1}e^{-(t-T)s}\norma{w_{1}}_{H^{2}(K)}\diff s\\
                                                                                            & \le C\int_{0}^{A+\varepsilon+\delta\ln{(1+r)}}(1+(r-is)^{2})^{-1}\norma{w_{1}}_{H^{2}(K)}\diff s\\
                                                                                            & =C\Biggl(\arctan{\biggl(r-i\Bigl(A+\varepsilon+\delta\ln{(1+r)}\Bigr)\biggr)}-\arctan{r}\Biggr)\norma{w_{1}}_{H^{2}(K)}\xrightarrow[r\to+\infty]{}0.
\end{split}
\end{equation*}
We have now proved that
\begin{equation*}
\rho U(t)\rho=\Pi_{A}(t)+E_{\omega}(t), \quad\textup{for}\quad w_{1}\in H^{2}(K)
\end{equation*}
($\omega$ is the curve $\omega=\Set{s-i\,\Bigl(A+\varepsilon+\delta\ln{(1+\abs{s})}\Bigr)|s\in\R}$).

We now show that for $t\gg 1$ and $w_{1}\in H^{2}(K)$ holds that
\begin{equation}
\label{boundEomegaR}
\norma{E_{\omega}(t)w_{1}}_{H^{2}(K)}\le Ce^{-tA}\norma{w_{1}}_{L^{2}(K)}.
\end{equation}
In order to do so, we use \cref{boundtruncatedresolvent} with $j=2$ and the assumption that there exist a compact set containing $\omega$ in which there are no poles of $R(z)$. We obtain
\begin{equation*}
\begin{split}
\norma{E_{\omega}(t)w_{1}}_{H^{2}(K)} & \le Ce^{-tA}\int_{-\infty}^{+\infty}e^{-t\delta\ln{(1+\abs{s})}}(1+\abs{s})e^{T\delta\ln{(1+\abs{s})}}\norma{w_{1}}_{L^{2}(K)}\diff s\\
                                                              & = Ce^{-tA}\int_{-\infty}^{+\infty}(1+\abs{s})^{-(t-T)\delta+1}\norma{w_{1}}_{L^{2}(K)}\diff s.
\end{split}
\end{equation*}
This last integral is finite if $-(t-T)\delta+1< -1$, which is equivalent to $t>T+\frac{2}{\delta}$.

In \cref{boundEomegaR} the above bound is in terms of the $L^{2}(K)$ norm of $w_{1}$ for $w_{1}\in H^{2}(K)$. Since, $H^{2}(K)$ is dense in $L^{2}(K)$, then the same bound holds for any $w_{1}\in L^{2}(K)$ and the theorem holds for any $w_{1}\in L^{2}(K)$.

The proof for arbitrary $w_{0}\in H_{\textup{comp}}^{1}(\R_{+}^{3})$ and $w_{1}=0$ follows the same steps, with the replacement $\sin{tz}/z\longrightarrow\cos{tz}$ in the formula for $w(t, x)$.}
\end{proof}
Now we turn to the Schr\"odinger equation. In this case, instead of giving the resonance expansion of the solution, we give the expansion of the kernel of the propagator, equivalent from the mathematical point of view and typically more useful in many applications in Quantum Mechanics. Previous early analysis of this kind has been given in \cite{AH-K:ResonanceExpansionGreenFunctionSchrodingerWave}, where the case of $\R^3$ with several point interactions or infinitely many point interactions periodically placed is treated. In the following theorem we consider the resonance expansion for the propagator of the absolutely continuous part of the Hamiltonian of a point interaction on the half-space. For the sake of simplicity in the exposition we exclude in the statement and proof the case where a zero energy eigenvalue is present. This, however, can be treated straightforwardly as already done for the wave equation.\\ The next Theorem treats the Dirichlet case, completely analogous result and proof hold true for Neumann boundary conditions.
\begin{proposition}
Let $H=-\Delta_{\alpha, y}^{\textup{D}}$ with $\alpha\ne -\frac{1}{8\pi y_{3}}$  and $H^+$ its projection on the absolutely continuous spectrum.
 Then there exists some $t_{0}(x, x', \alpha, y)$ such that for $t>t_{0}$ the following expansion holds
\begin{equation}
\label{resonanceexpansion}
\begin{split}
e^{-itH_{+}}(x, x')&=(2\pi it)^{-\frac{3}{2}}e^{\frac{i}{2t}\abs{x-x'}^{2}}+\sum_{n}2z_{n}e^{-itz_{n}^{2}}\Res_{z=z_{n}}{\biggl(\Bigl(\Gamma^{\textup D}_{\alpha, y}(z)\Bigr)^{-1}\biggr)}G^{\textup D}_{z_{n}, y}(x)G^{\textup D}_{z_{n}, y}(x')\\ &-\frac{1}{2\pi}\int_{0}^{+\infty}\biggl(\hat{R}\Bigl(e^{-i\frac{\pi}{4}}s\Bigr)-\hat{R}\Bigl(-e^{-i\frac{\pi}{4}}s\Bigr)\biggr)e^{-ts}\diff s,
\end{split}
\end{equation}
with $\hat{R}(z)=(H-z^{2})^{-1}-(-\Delta^D-z^{2})^{-1}$. The sum in \cref{resonanceexpansion} is taken over all resonances $z_{n}$ with $-\frac{\pi}{4}<\Arg{z_{n}}<0$.
\end{proposition}
\begin{proof}
By the definition of $\hat{R}$ and functional calculus it holds
\begin{equation*}
e^{-itH_{+}}(x, x')=-\frac{1}{2\pi i}\int_{0}^{+\infty}e^{-itz^{2}}\hat{R}(z)(x, x')2z\diff z-\frac{1}{2\pi i}\int_{0}^{+\infty}e^{-itz^{2}}G_{z, y}^{\textup{D}}(x, x')2z\diff z.
\end{equation*}
It holds that
\begin{equation*}
\frac{1}{2\pi i}\int_{0}^{+\infty}e^{-itz^{2}}G_{z, y}^{\textup{D}}(x, x')2z\diff z=(2\pi it)^{-\frac{3}{2}}e^{i\abs{x-x'}^{2}{2t}}.
\end{equation*}
Instead, because of the parity of $\hat{R}(z)$, we can rewrite the other integral as
\begin{equation*}
-\frac{1}{2\pi i}\int_{0}^{+\infty}e^{-itz^{2}}\Bigl(\hat{R}(z)-\hat{R}(-z)\Bigr)(x, x')2z\diff z.
\end{equation*}
Using the residue theorem, this integral can be decomposed as
\begin{equation}
\begin{split}
\label{residueCRgamma}
\frac{1}{2\pi i}\int_{C_{R}}e^{-itz^{2}}\Bigl(\hat{R}(z)-\hat{R}(-z)\Bigr)(x, x')2z\diff z\ +\ &\frac{1}{2\pi i}\int_{\gamma_{R}}e^{-itz^{2}}\Bigl(\hat{R}(z)-\hat{R}(-z)\Bigr)(x, x')2z\diff z\\
+\ &\sum_{n}2z_{n}e^{-itz_{n}^{2}}\Res_{z=z_{n}}{\Bigl(\hat{R}(z)-\hat{R}(-z)\Bigr)}(x, x'),
\end{split}
\end{equation}
where $C_{R}$ is the arc of circumference of radius $R$ going clockwise from $z=R$ to $z=Re^{-i\pi/4}$, $\gamma_{R}$ is the segment going from $z=Re^{-i\pi/4}$ towards the origin and the sum is over the $z_{n}$ in the region delimited by $C_{R}$, $\gamma_{R}$ and the real positive semi-axis.

We prove that the integral over $C_{R}$ vanishes for $R\to +\infty$. We prove this just for the term with $\hat{R}(z)$. Analogous observations prove the same for $\hat{R}(-z)$. By  \cref{DirSpec}, there are infinite resonances $z_{n}$ and $\abs{z_{n}}\to +\infty$ for $n\to +\infty$. This means that, in order to study the limit $R\to +\infty$, we have to distinguish the cases
\begin{itemize}
\item $R\ne\abs{z_{n}}$, for all $n\in\N$;
\item $R=\abs{z_{n}}$, for some $n\in\N$.
\end{itemize}
In the former case, the integral is given by
\begin{equation}
\label{integralarcnores}
\int_{C_{R}}e^{-itz^{2}}\biggl(\alpha-\frac{iz}{4\pi}+\frac{e^{2izy_{3}}}{8\pi y_{3}}\biggr)^{-1}\frac{e^{iz\abs{x-y}}}{4\pi\abs{x-y}}\frac{e^{iz\abs{x'-y}}}{4\pi\abs{x'-y}}2z\diff z.
\end{equation}
We substitute $z=Re^{i\theta}$ and estimate the absolute value of the integrand
\begin{equation*}
\abs{e^{-itR^{2}e^{2i\theta}}\biggl(\alpha-\frac{iRe^{i\theta}}{4\pi}+\frac{e^{2iRe^{i\theta}y_{3}}}{8\pi y_{3}}\biggr)^{-1}\frac{e^{iRe^{i\theta}\abs{x-y}}}{4\pi\abs{x-y}}\frac{e^{iRe^{i\theta}\abs{x'-y}}}{4\pi\abs{x'-y}}2ie^{i\theta}R^{2}}=\abs{I}.
\end{equation*}
The second factor is bounded, because the path does not cross any resonance, hence it holds that
\begin{equation*}
\abs{I}\le C(x, x', y)R^{2}e^{R\sin{\theta}(2tR\cos{\theta}-\abs{x-y}-\abs{x'-y})}.
\end{equation*}
Since we are in the complex lower half-plane, $\sin{\theta}<0$. So if the other factor in the exponent is positive, then in the limit the integral vanishes. By recalling that $-\frac{\pi}{4}<\theta<0$, we have that
\begin{equation*}
2tR\cos{\theta}-\abs{x-y}-\abs{x'-y}>\sqrt{2}Rt-\abs{x-y}-\abs{x'-y},
\end{equation*}
which is positive for $t>\frac{\sqrt{2}}{2R}\Bigl(\abs{x-y}+\abs{x'-y}\Bigr)$ and hence the integral vanishes in the limit. 

In the other case, the path of integration has to be modified because the previous path would pass across the pole $z_{n}$. So, the integral \cref{integralarcnores} has to be replaced by the following sum of integrals of the same function
\begin{equation*}
-\frac{1}{2\pi i}\int_{C_{R}^{1}}\dots\diff z-\frac{1}{2\pi i}\int_{C_{\varepsilon}}\dots\diff z-\frac{1}{2\pi i}\int_{C_{R}^{2}}\dots\diff z,
\end{equation*}
where $C_{\varepsilon}$ is the semicircle of radius $\varepsilon$ centered at $z_{n}$, $C_{R}^{1}$ is the arc of circle of radius $R$ centered at the origin starting at $Re^{-i\frac{\pi}{4}}$ and ending when intersecting $C_{\varepsilon}$, and $C_{R}^{2}$ is the arc of circle of radius $R$ centered at the origin between the point $R$ and the intersection with $C_{\varepsilon}$. The path is run clockwise. The integrals on $C_{R}^{1}$ and $C_{R}^{2}$ are proven to vanish in the limit $n\to +\infty$ for $t$ large enough in the same way as in the first case. We consider the remaining integral, it holds that
\begin{equation*}
-\frac{1}{2\pi i}\int_{C_{\varepsilon}}\dots\diff z=-\frac{1}{2}\Res_{z=z_{n}}{\Biggl(e^{-itz^{2}}\biggl(\alpha-\frac{iz}{4\pi}+\frac{e^{2izy_{3}}}{8\pi y_{3}}\biggr)^{-1}\frac{e^{iz\abs{x-y}}}{4\pi\abs{x-y}}\frac{e^{iz\abs{x'-y}}}{4\pi\abs{x'-y}}2z\Biggr)}.
\end{equation*}
In order to find the residue, we study the divergent term
\begin{equation*}
\Bigl(\Gamma_{\alpha, y}^{\textup{D}}(z)\Bigr)^{-1}=\biggl(\alpha-\frac{iz}{4\pi}+\frac{e^{2izy_{3}}}{8\pi y_{3}}\biggr)^{-1}.
\end{equation*}
Using the definition of $z_{n}$ we have that
\begin{equation*}
\begin{split}
\Bigl(\Gamma_{\alpha, y}^{\textup{D}}(z)\Bigr) & = \alpha-\frac{i}{4\pi}(z-z_{n})-\frac{iz_{n}}{4\pi}+\frac{e^{2iy_{3}(z-z_{n})}}{8\pi y_{3}}e^{2iy_{3}z_{n}}\\
                                                                                                & = \frac{e^{2iy_{3}z_{n}}}{8\pi y_{3}}\biggl(e^{2iy_{3}(z-z_{n})}-1\biggr)-\frac{i}{4\pi}(z-z_{n})\\
                                                                                                & = \frac{e^{2iy_{3}z_{n}}}{8\pi y_{3}}\sum_{k=1}^{+\infty}\frac{2^{k}i^{k}y_{3}^{k}}{k!}(z-z_{n})^{k}-\frac{i}{4\pi}(z-z_{n})\\
                                                                                                & = \frac{i}{4\pi}\biggl(e^{2iy_{3}z_{n}}-1\biggr)(z-z_{n})+O((z-z_{n})^{2}).
\end{split}
\end{equation*}
This means that
\begin{equation*}
\Res_{z=z_{n}}{\biggl(\Bigl(\Gamma_{\alpha, y}^{\textup{D}}(z)\Bigr)^{-1}\biggr)}=\lim_{z\to z_{n}}\frac{z-z_{n}}{\alpha-\frac{iz}{4\pi}+\frac{e^{2iy_{3}z_{n}}}{8\pi y_{3}}}=\frac{4\pi i}{1-e^{2iy_{3}z_{n}}}.
\end{equation*}
The denominator would vanish only if $z_{n}=\frac{k\pi}{y_{3}}$ for some $k\in\Z$. In order for this value to be a resonance, it has to solve
\begin{equation*}
\frac{ik}{4y_{3}}=\alpha+\frac{1}{8\pi y_{3}},
\end{equation*}
but the the left side is purely imaginary, while the right one is real. So the equation holds only if both sides vanish, which happens if and only if $z_{n}=0$ and $\alpha=-\frac{1}{8\pi y_{3}}$.
Hence, if we state $z_{n}=r_{n}e^{i\theta_{n}}$, we have
\begin{equation*}
\begin{split}
\abs{-\frac{1}{2\pi i}\int_{C_{\varepsilon}}\dots\diff z} & = \abs{e^{-itz_{n}^{2}}\frac{4\pi i}{e^{2iy_{3}z_{n}}-1}\frac{e^{iz_{n}\abs{x-y}}}{4\pi\abs{x-y}}\frac{e^{iz_{n}\abs{x'-y}}}{4\pi\abs{x'-y}}z_{n}}\\
                                                                                & \le C'(x, x', y) e^{r_{n}\sin{\theta_{n}}(2tr_{n}\cos{\theta_{n}}-\abs{x-y}-\abs{x'-y})}.
\end{split}
\end{equation*}
Also in this case
\begin{equation*}
2tr_{n}\cos{\theta_{n}}-\abs{x-y}-\abs{x'-y}\ge \sqrt{2}tr_{n}-\abs{x-y}-\abs{x'-y},
\end{equation*}
with the last expression being positive for $t>\frac{\sqrt{2}}{2r_{n}}\Bigl(\abs{x-y}+\abs{x'-y}\Bigr)$. So, in the limit $n\to +\infty$ also the integral on $C_{\varepsilon}$ vanishes for $t$ large enough.

Going back to \cref{residueCRgamma}, we consider the sum in it. First we observe that $\hat{R}(-z)$ has no poles in the region considered, because all its poles are in the upper complex plane. The expression of $(-\Delta_{\alpha, y}^{\textup{D}}-z^{2})^{-1}$ and the fact that for $R\to +\infty$ the region of interest becomes $-\frac{\pi}{4}<\Arg{z}<0$ let us recollect the expression present in \cref{resonanceexpansion}.

The integral over $\gamma_{R}$ can be expressed, with the substitution $z=e^{-i\pi/4}\sqrt{s}$, as follows
\begin{equation*}
\frac{1}{2\pi i}\int_{0}^{+\infty}e^{-ts}\biggl(\hat{R}\Bigl(e^{-i\frac{\pi}{4}}\sqrt{s}\Bigr)-\hat{R}\Bigl(-e^{-i\frac{\pi}{4}}\sqrt{s}\Bigr)\biggr)(x, x')2e^{-i\frac{\pi}{4}}\sqrt{s}e^{-i\frac{\pi}{4}}\frac{\diff s}{2\sqrt{s}},
\end{equation*}
which returns the missing term. 
\end{proof}

\section*{Appendix - Proof of \cref{definition:onepointD}}
Here we construct the self-adjoint extensions $-\Delta_{\alpha, y}^{\textup{D}}$ of $-\mathring\Delta^{\textup{D}}$ using von Neumann theory of extensions and boundary value spaces. A very similar construction (not here reported) applies for the Neumann operator. See also, for a different analysis in the context of the description of Casimir effect, the paper \cite{ACSZ08}.
 
In \cref{preliminaries} we noted that $-\mathring\Delta_y^{\textup{D}}$ is a closed symmetric operator. If its deficiency indices $N_{\pm}\left(-\mathring\Delta_y^{\textup{D}}\right)=\dim{\ran{\left(-\mathring\Delta_y^{\textup{D}}\pm i\right)}}$ are equal then it admits an $N_{\pm}^{2}$-parameters family of self-adjoint extensions. The following lemma characterizes the deficiency spaces.
\begin{lemma}
Let $\mathcal{H}_{\pm}=\ran{(-\mathring\Delta_y^{\textup{D}}\pm i)}^{\perp}$ be the deficiency subspaces related to the operator $-\mathring\Delta_y^{\textup{D}}$;
\begin{equation*}
\mathcal{H}_{\pm}=\Span{\Set{G_{\sqrt{\mp i}, y}^{\textup{D}}}}
\end{equation*}
and then  then $N_{\pm}\left(-\mathring\Delta_y^{\textup{D}}\right)=1$.
\end{lemma}
\begin{proof}
$\psi\in\mathcal{H}_{\pm}$ implies
\begin{equation*}
\begin{cases}
\langle\psi, (-\Delta\pm i)\varphi\rangle_{L^{2}(\R_{+}^{3})}=0\\
\psi\in L^{2}(\R_{+}^{3})
\end{cases}
,\quad\forall\varphi\in D(-\mathring\Delta^{\textup{D}}).
\end{equation*}
But since, $\varphi(y)=0$, we have that
\begin{equation*}
\langle G_{\sqrt{\mp i}, y}^{\textup{D}}, (-\Delta\pm i)\varphi\rangle_{L^{2}(\R_{+}^{3})}=\langle(-\Delta\mp i)G_{\sqrt{\mp i}, y}^{\textup{D}}, \varphi\rangle_{L^{2}(\R_{+}^{3})}=\varphi(y)=0,
\end{equation*}
which means that $\Span{\Set{G_{\sqrt{\mp i}, y}^{\textup{D}}}}\subseteq\mathcal{H}_{\pm}$. Let $\psi_{\pm}\in\ran{(-\mathring\Delta_y^{\textup{D}}\pm i)}^{\perp}$ and $\varphi\in D(-\mathring\Delta_y^{\textup{D}})$. Then for some $c^{\pm}$ independent of $\varphi$ it holds that
\begin{equation}
\label{psipmortogonal}
\langle\psi_{\pm}, (-\Delta\pm i)\varphi\rangle_{L^{2}(\R_{+}^{3})}=c^{\pm}\varphi(y).
\end{equation}
In fact, let
\begin{equation}
\label{phitilde}
\tilde{\varphi}=\varphi-\varphi(y)\chi_{j},
\end{equation}
where $\chi_{j}\in C_{0}^{\infty}(\R_{+}^{3})$, $\chi(y)=1$. Then $\tilde{\varphi}\in D(-\mathring\Delta_y^{\textup{D}})$. Substituting \cref{phitilde} in \cref{psipmortogonal} and using this fact we obtain that $c^{\pm}=(\psi_{\pm}, (-\Delta\pm i)\chi_{j})$. But the constants $c^{\pm}$ are uniquely determined by $\psi_{\pm}$. In fact, let $\tilde{\psi_{\pm}}$ be such that
\begin{equation*}
\langle\tilde{\psi}_{\pm}, (-\Delta\pm i)\varphi\rangle_{L^{2}(\R_{+}^{3})}=c^{\pm}\varphi(y).
\end{equation*}
Then $\langle(\psi_{\pm}-\tilde{\psi}_{\pm}), (-\Delta\pm i)\varphi\rangle_{L^{2}(\R_{+}^{3})}=0$ $\forall\varphi\in D(-\mathring\Delta_y^{\textup{D}})$, which implies that $\psi_{\pm}=\tilde{\psi}_{\pm}$.

Finally, we observe that
\begin{equation*}
\psi_{\pm}=c^{\pm}G_{\sqrt{\mp i}, y}^{\textup{D}}
\end{equation*}
satisfies \cref{psipmortogonal}, thereby proving $\mathcal{H}_{\pm}\subseteq\Span{\Set{G_{\sqrt{\mp i}, y}^{\textup{D}}}}$.
\end{proof}
Now in order to construct the family of extensions of $-\mathring\Delta_y^{\textup{D}}$ we will use the theory of boundary value spaces. On this topic we recall only the results essential for our analysis and refer to \cite{GG:BVPOperatorDifferentialEquations} and \cite{BFM:PointInteractionBoundedDomains} for additional details.
\begin{definition}[Boundary Value Space]
Let $A$ be a densely defined, closed symmetric operator in an Hilbert space $\mathscr{H}$ with equal, finite or infinite, deficiency indices. A triple $(\mathscr{V}, \Gamma_{1}, \Gamma_{2})$, where $\mathscr{V}$ is a Hilbert space and $\Gamma_{i}: D(A^{*})\to\mathscr{V}$ $i=1, 2$ are bounded linear operators, is called a boundary value space of the operator $A$ if
\begin{gather}
(\psi, A^{*}\varphi)_{\mathscr{H}}-(A^{*}\psi, \varphi)_{\mathscr{H}}=(\Gamma_{1}\psi, \Gamma_{2}\varphi)_{\mathscr{V}}-(\Gamma_{2}\psi, \Gamma_{1}\varphi)_{\mathscr{V}},\quad\forall\psi, \varphi\in D(A^{*})\\
\textup{the map}\quad(\Gamma_{1}, \Gamma_{2}): D(A^{*})\to\mathscr{V}\oplus\mathscr{V}\quad\textup{is surjective}.
\end{gather}
\end{definition}
Now the following proposition holds.
\begin{theorem}
\label{boundaryvaluespacematrices}
Let $A$ be a densely defined, closed symmetric operator in $\mathscr{H}$ with equal deficiency  $(N, N)$, $N<+\infty$ and let $(\mathscr{V}, \Gamma_{1}, \Gamma_{2})$ be its boundary value space. Let 
\begin{equation*}
W=\Set{E=\begin{pmatrix} B\, C\end{pmatrix}| BC^{*}=CB^{*}, \ran{E}=N},
\end{equation*}
where $B, C$ are complex $N\times N$ matrices and $E=\begin{pmatrix} B\, C\end{pmatrix}$ denotes the $N\times 2N$ matrix obtained by horizontal juxtaposition of $B$ and $C$.
There is a bijective correspondence between the self-adjoint extensions of $A$ and the set $W$ defined above. A self-adjoint extension $A^{B, C}$, corresponding to $\begin{pmatrix} B\, C\end{pmatrix}\in W$, is given by the restriction of $A^{*}$ to those elements $\psi\in D(A^{*})$ satisfying the boundary conditions
\begin{equation}
\label{operatorialbc}
B\Gamma_{1}\psi=C\Gamma_{2}\psi.
\end{equation}
\end{theorem}
\begin{remark}
When $N=1$, the set $W$ corresponding to self-adjoint extensions of the operator in exam is made of vectors $E=(z, w)$ with $z, w\in\C$ such that $z\overline{w}=w\overline{z}$ with $\abs{E}>0$. So the possible cases are $z=0 \land w\neq 0$, $z\neq 0\land w=0$ or $z=\rho e^{i\theta}\land w=re^{i\theta}$ for some $\rho, r>0$ and $\theta\in[0, 2\pi)$. In this latter case, \cref{operatorialbc} takes the form of $\rho e^{i\theta}\Gamma_{1}\psi=re^{i\theta}\Gamma_{2}\psi$ which means that we can assume without loss of generality $z=\alpha\in\R$ and $w=1$ and \cref{operatorialbc} takes the form
\begin{equation}
\label{operatorialbc1point}
\alpha\Gamma_{1}\psi=\Gamma_{2}\psi
\end{equation}
\end{remark}
So in order to construct the sought self-adjoint extensions it is sufficient to find a boundary value space for $H_{0}=-\mathring\Delta_y^{\textup{D}}$. By the definition of $H_{0}$ it is straightforward that
\begin{equation*}
D(H_{0}^{*})=\left\{\psi\in H_{\textup{loc}}^{2}\left(\R_{+}^{3}\setminus\{y\}\right)\cap L^{2}\left(\R_{+}^{3}\right): \Delta\psi\in L^{2}\left(\R_{+}^{3}\right)\right\}.
\end{equation*}
Consider the following operators
\begin{align}
\label{Gamma1}
&\Gamma_{1}: D(H_{0}^{*})\to \C\qquad\Gamma_{1}\psi=\lim_{x\to y}4\pi\abs{x-y}\psi(x),\\
&\label{Gamma2}
\Gamma_{2}: D(H_{0}^{*})\to \C\qquad\Gamma_{2}\psi=\lim_{x\to y}\Biggl(\psi(x)-\frac{\Gamma_{1}\psi}{4\pi\abs{x-y}}\Biggr).
\end{align}
We claim that
\begin{theorem}
The triple $(\C, \Gamma_{1}, \Gamma_{2})$ defined by \cref{Gamma1} and \cref{Gamma2} forms a boundary value space for $H_{0}$.
\end{theorem}
\begin{proof}
By the von Neumann formula, we can write two generic vectors $\psi, \varphi\in D(H_{0}^{*})$ as follows
\begin{equation}
\label{domaindecomposition}
\psi=\psi_{0}+aG_{\sqrt{i}, y}^{\textup{D}}+bG_{\sqrt{-i}, y}^{\textup{D}}, \quad \varphi=\varphi_{0}+\alpha G_{\sqrt{i}, y}^{\textup{D}}+\beta G_{\sqrt{-i}, y}^{\textup{D}},
\end{equation}
where $\psi_{0}, \varphi_{0}\in D(H_{0})$ and $a, b, \alpha, \beta\in\C$. We now consider
\begin{equation}
\label{leftside}
\langle\psi, H_{0}^{*}\varphi\rangle_{L^{2}(\R_{+}^{3})}-\langle H_{0}^{*}\psi, \varphi\rangle_{L^{2}(\R_{+}^{3})}.
\end{equation}
By some computations, the definition of $G_{\sqrt{\mp i}, y}^{\textup{D}}$ and fact that $\psi_{0}, \varphi_{0}\in D(H_{0})$, it holds that (we omit the subscript of the inner product for brevity)
\begin{equation}
\label{leftside2}
\langle\psi, H_{0}^{*}\varphi\rangle-\langle H_{0}^{*}\psi, \varphi\rangle=2i\Bigl(a\overline{\alpha_{l}}-b_{l}\overline{\beta}\Bigr)\langle G_{\sqrt{i}, y}^{\textup{D}}, G_{\sqrt{i}, y}^{\textup{D}}\rangle.
\end{equation}
For convenience we fix some $y'\in\R_{+}^{3}\setminus\{y\}$. By the definition of $G_{\sqrt{\mp i}, y}^{\textup{D}}$, we have
\begin{equation*}
\begin{split}
\langle G_{\sqrt{i}, y}^{\textup{D}}, G_{\sqrt{i}, y'}^{\textup{D}}\rangle & = -i\,\langle iG_{\sqrt{i}, y}^{\textup{D}}, G_{\sqrt{i}, y'}^{\textup{D}}\rangle=-i\,\langle H_{0}G_{\sqrt{i}, y}^{\textup{D}}, G_{\sqrt{i}, y'}^{\textup{D}}\rangle+i\overline{G_{\sqrt{i}, y'}^{\textup{D}}(y)}\\
                                                                                                             & = -i\,\langle G_{\sqrt{i}, y}^{\textup{D}}, H_{0}G_{\sqrt{i}, y'}^{\textup{D}}\rangle+iG_{\sqrt{-i}, y'}^{\textup{D}}(y)\\
                                                                                                             & = -\langle G_{\sqrt{i}, y}^{\textup{D}}, G_{\sqrt{i}, y'}^{\textup{D}}\rangle-iG_{\sqrt{i}, y}^{\textup{D}}(y')+iG_{\sqrt{-i}, y'}^{\textup{D}}(y),
\end{split}
\end{equation*}
equivalent to
\begin{equation}
\label{Gknel}
\langle G_{\sqrt{i}, y}^{\textup{D}}, G_{\sqrt{i}, y'}^{\textup{D}}\rangle = -\frac{i}{2}\Bigl(G_{\sqrt{i}, y}^{\textup{D}}(y')-G_{\sqrt{-i}, y'}^{\textup{D}}(y)\Bigr).
\end{equation}
Now, if $y'\to y$, both terms in the difference in \cref{Gknel} diverge, but the quantity has still a finite limit. We rewrite last equation using the definition for $G_{\sqrt{\mp i}, y}^{\textup{D}}$
\begin{equation*}
\langle G_{\sqrt{i}, y}^{\textup{D}}, G_{\sqrt{i}, y'}^{\textup{D}}\rangle = -\frac{i}{2}\Bigl(G_{\sqrt{i}, y}^{0}(y')-G_{\sqrt{-i}, y'}^{0}(y)-h_{\sqrt{i}, y}^{\textup{D}}(y')+h_{\sqrt{-i}, y'}^{\textup{D}}(y)\Bigr).
\end{equation*}
The last two terms are finite for $y=y'$ and we can write
\begin{equation*}
h_{\sqrt{-i}, y}^{\textup{D}}(y)-h_{\sqrt{i}, y}^{\textup{D}}(y)=h_{\sqrt{-i}, y}^{\textup{D}}(y)-\overline{h_{\sqrt{-i}, y}^{\textup{D}}(y)}=2i\Im{h_{\sqrt{-i}, y}^{\textup{D}}(y)}.
\end{equation*}
For the first two terms instead it holds
\begin{equation*}
G_{\sqrt{i}, y}^{0}(y')-G_{\sqrt{-i}, y'}^{0}(y)=G_{\sqrt{i}, y}^{0}(y')-G_{\sqrt{-i}, y}^{0}(y').
\end{equation*}
So, to evaluate the expression  on the left side, we can take the limit $y'\to y$ of the right one
\begin{equation*}
\lim_{y'\to y}\Bigl(G_{\sqrt{i}, y}^{0}(y')-G_{\sqrt{-i}, y}^{0}(y')\Bigr)=\lim_{y'\to y}\frac{e^{-\sqrt{-i}\abs{y'-y}}-e^{-\sqrt{i}\abs{y'-y}}}{4\pi\abs{y'-y}}=\frac{\sqrt{i}-\sqrt{-i}}{4\pi}.
\end{equation*}
So we found that
\begin{equation*}
\langle G_{\sqrt{i}, y}^{\textup{D}}, G_{\sqrt{i}, y}^{\textup{D}}\rangle_{L^{2}(\R_{+}^{3})}=\Im{h_{\sqrt{-i}, y}^{\textup{D}}(y)}+\frac{i}{8\pi}(\sqrt{-i}-\sqrt{i}).
\end{equation*}
Then \cref{leftside2} can be rewritten as
\begin{equation}
\label{leftside3}
\langle\psi, H_{0}^{*}\varphi\rangle-\langle H_{0}^{*}\psi, \varphi\rangle=-\frac{1}{4\pi}\Bigl(a\overline{\alpha}-b\overline{\beta}\Bigr)\left(\sqrt{-i}-\sqrt{i}\right)-\Bigl(a\overline{\alpha}-b\overline{\beta}\Bigr)\Bigl(h_{\sqrt{i}, y}^{\textup{D}}(y)-h_{\sqrt{-i}, y}^{\textup{D}}(y)\Bigr).
\end{equation}

We consider now the operators $\Gamma_{1}$ and $\Gamma_{2}$ acting on $D(H_{0}^{*})$
\begin{gather}   
\Gamma_{1}\psi=a+b\\
\Gamma_{2}\psi=-\frac{1}{4\pi}(a\sqrt{-i}+b\sqrt{i})-ah_{\sqrt{i}, y}^{\textup{D}}(y)-bh_{\sqrt{-i}, y}^{\textup{D}}(y).
\end{gather}
Substituting these expressions into $\langle\Gamma_{1}\psi, \Gamma_{2}\varphi\rangle_{\C}-\langle\Gamma_{2}\psi, \Gamma_{1}\varphi\rangle_{\C}$ leads to \cref{leftside3} completing the proof.
\end{proof}
\begin{proof}[Proof of \cref{definition:onepointD}]
By the von Neumann decomposition formula, we can write
\begin{equation}
\label{representation1}
\psi=\psi_{0}+aG_{\sqrt{i}, y}^{\textup{D}}+bG_{\sqrt{-i}, y}^{\textup{D}}, \quad\psi_{0}\in D(H_{0}).
\end{equation}
We fix $z\in\C$ with $\Im z>0$ and set
\begin{gather}
\varphi^{z}=\psi_{0}+aG_{\sqrt{i}, y}^{\textup{D}}+bG_{\sqrt{-i}, y}^{\textup{D}}-qG_{z, y}^{\textup{D}}\\
q=\Gamma_{1}\psi=a+b
\end{gather}
We show that $\varphi^{z}\in D(-\Delta^{\textup{D}})$. The Dirichlet boundary condition is satisfied because $\varphi^{z}$ is a sum of terms satisfying it. We then need to prove that $\varphi^{z}\in H^{2}(\R_{+}^{3})$. The decay at infinity is exponentially fast and so we need to check only the possible singularity in $y$. In a neighborhood of $y$ the behavior of the function $\varphi^{z}$ is controlled by the one of
\begin{equation*}
f(x)=a\frac{e^{-\sqrt{-i}\abs{x-y}}}{4\pi\abs{x-y}}+b\frac{e^{-\sqrt{i}\abs{x-y}}}{4\pi\abs{x-y}}-q\frac{e^{iz\abs{x-y}}}{4\pi\abs{x-y}}.
\end{equation*}
It is easy to verify that $f\in H^{1}(B_{r}(y))$ for $r>0$ small enough. Also $\Delta f\in H^{2}(B_{r}(y))$. In fact, using the definition of $G_{\cdot, \cdot}^{\textup{D}}$ we have
\begin{equation*}
\Delta f(x)=-ia\frac{e^{-\sqrt{-i}\abs{x-y}}}{4\pi\abs{x-y}}+ib\frac{e^{-\sqrt{i}\abs{x-y}}}{4\pi\abs{x-y}}-zq\frac{e^{-\sqrt{z}\abs{x-y}}}{4\pi\abs{x-y}}+(-a-b+q)\delta(x-y),
\end{equation*}
but the distributional term is cancelled due to $q=a+b$. So $\varphi^{z}\in D(H_{0})$ and the representations \cref{representation1} and 
\begin{equation}
\label{representation2}
\psi=\varphi^{z}+qG_{z, y}^{\textup{D}}, \quad \varphi^{z}\in D(-\Delta^{\textup{D}})
\end{equation}
are equivalent.

The value $\varphi^{z}(y)$ is linked to $\Gamma_{1}\psi$ and $\Gamma_{2}\psi$. In fact
\small
\begin{equation}
\label{Gammaproof}
\varphi^{z}(y)=\lim_{x\to y}\Bigl(aG_{\sqrt{i}, y}^{\textup{D}}(x)+bG_{\sqrt{-i}, y}^{\textup{D}}(x)-qG_{z, y}^{\textup{D}}(x)\Bigr)=\Gamma_{2}\psi-q\frac{iz}{4\pi}+qh_{z, y}^{\textup{D}}(y)=\left(\alpha-\frac{iz}{4\pi}+h_{z, y}^{\textup{D}}(y)\right)q=\Gamma_{\alpha, y}^{\textup{D}}q
\end{equation}
\normalsize
and so the representation of $D(-\Delta_{\alpha, y}^{\textup{D}})$ in \cref{dirichletoperator} is proved.

Let $\psi\in D\left(-\Delta_{\alpha, y}^{\textup{D}}\right)$ and $\varphi\in D(H_{0})$. In order to prove the formula for the action of $-\Delta_{\alpha, y}^{\textup{D}}$ we consider the inner product
\begin{equation*}
\begin{split}
\langle-\Delta_{\alpha, y}^{\textup{D}}\psi, \varphi\rangle_{L^{2}(\textup{D})} & = \langle-\Delta_{\alpha, y}^{\textup{D}}\Bigl(\varphi^{z}+qG_{z, y}^{\textup{D}}\Bigr), \varphi\rangle_{L^{2}(\Omega)}=\langle H_{0}^{*}\varphi^{z}, \varphi\rangle+\sum_{k=1}^{N}q\langle H_{0}^{*}G_{z, y}^{\textup{D}}, \varphi\rangle\\
                                                           & = \langle H_{0}\varphi^{z},\varphi\rangle+\sum_{k=1}^{N}qz^{2}\langle G_{z, y}^{\textup{D}}, \varphi\rangle.
\end{split}
\end{equation*}
This equation being valid $\forall\varphi\in D(H_{0})$ proves the action formula.

Let $\varphi\in L^{2}(\R_{+}^{3})$. Due to  Kre\u{\i}n's formula the resolvent has the form
\begin{equation}
\label{kreinHAB}
\left(-\Delta_{\alpha, y}^{\textup{D}}-z^{2}\right)^{-1}\varphi=\left(-\Delta^{\textup{D}}-z^{2}\right)^{-1}\varphi+q(\varphi)G_{z, y}^{\textup{D}},
\end{equation}
with $q(\varphi)$ to be determined. Equation \cref{Gammaproof} with $\varphi^{z}(y)$ replaced by $\left(-\Delta^{\textup{D}}-z^{2}\right)^{-1}\varphi$ gives the desired result
\begin{equation*}
q=\left(\Gamma_{\alpha, y}^{\textup{D}}\right)^{-1}\left[\left(-\Delta^{\textup{D}}-z^{2}\right)^{-1}\varphi\right](y).
\end{equation*}
\end{proof}

\section*{Acknowledgments.}
\noindent The first author acknowledges the support of the Next Generation EU - Prin 2022 project "Singular Interactions and Effective Models in Mathematical Physics- 2022CHELC7". The authors thank the anonymous referee for several relevant remarks and corrections. The first author is also grateful to Prof.Tanya Christiansen for discussions and elucidations about the results in \cite{CDG25}.

\end{document}